\documentclass{IEEEtran}
\usepackage{layouts} 
\DeclareMathAlphabet{\mathsfb}{OT1}{cmss}{bx}{n}
\usepackage{amsmath}
\usepackage{amsfonts}

\usepackage[ruled,vlined,linesnumbered]{algorithm2e}
\usepackage{graphicx}
\usepackage{hyperref}
\newcommand{\Arikan}{Ar\i kan}

\newcommand{\myprobb}[1]{\mathbb{P}(#1)}
\newcommand{\myset}[1]{\left\{#1\right\}}
\newcommand{\mysett}[1]{\{#1\}}

\newcommand{\myrmax}[1]{\max \left(#1\right)}

\newcommand{\mycount}[1]{|#1|}

\newcommand{\bfu}{\mathbf{u}}
\newcommand{\sfbu}{\mathsfb{u}}
\newcommand{\sfbU}{\mathsfb{U}}
\newcommand{\sfu}{\mathsf{u}}
\newcommand{\sfU}{\mathsf{U}}
\newcommand{\sfbc}{\mathsfb{c}}
\newcommand{\sfc}{\mathsf{c}}
\newcommand{\sfby}{\mathsfb{y}}
\newcommand{\sfbY}{\mathsfb{Y}}
\newcommand{\sfy}{\mathsf{y}}
\newcommand{\bfy}{\mathbf{y}}

\newcommand{\subscriptodd}{\mathrm{odd}}
\newcommand{\subscripteven}{\mathrm{even}}
\newcommand{\xor}{\oplus}
\newcommand{\calY}{\mathcal{Y}}
\newcommand{\calX}{\mathcal{X}}
\newcommand{\phaseAndBranch}[2]{\langle #2 \rangle_{#1}}
\newcommand{\arrayPhaseAndBranch}[3]{#1_{#2}[\langle #3 \rangle]}
\newcommand{\arrayAndBranch}[3]{#1_{#2}[#3]}
\newcommand{\layerVar}{\lambda}
\newcommand{\layerCapitalVar}{\Lambda}
\newcommand{\phaseVar}{\varphi}
\newcommand{\branchVar}{\beta}
\newcommand{\floor}[1]{\lfloor #1 \rfloor}
\newcommand{\floorr}[1]{\left\lfloor #1 \right\rfloor}

\newcommand{\genericReceived}{y}
\newcommand{\bfgenericReceived}{\bfy}
\newcommand{\numberCalls}[2]{\#^{(#1)}_{#2}}

\newtheorem{theo}{Theorem}
\newtheorem{lemm}[theo]{Lemma}
\newtheorem{defi}{Definition}

\newcommand{\twobibs}[2]{#2} 

\begin{document}
\title{List Decoding of Polar Codes}
\author{\authorblockN{Ido Tal \qquad Alexander Vardy}\\
\authorblockA{
University of California San Diego,\\
La Jolla, CA 92093, USA\\
Email: {\tt idotal@ieee.org}, {\tt avardy@ucsd.edu}
}}
\maketitle
\SetAlFnt{\small}
\SetKwInOut{Require}{Require} 
\SetKwFunction{getArrayPointerP}{getArrayPointer\_P}
\SetKwFunction{getArrayPointerC}{getArrayPointer\_C}
\SetKwFunction{assignInitialPath}{assignInitialPath}
\SetKwFunction{pathIndexInactive}{pathIndexInactive}
\SetKwFunction{continuePathsFrozenBit}{continuePaths\_FrozenBit}
\SetKwFunction{continuePathsUnfrozenBit}{continuePaths\_UnfrozenBit}
\SetKwFunction{recursivelyCalcP}{recursivelyCalcP}
\SetKwFunction{recursivelyUpdateB}{recursivelyUpdateB}
\SetKwFunction{recursivelyUpdateC}{recursivelyUpdateC}
\SetKwFunction{killPath}{killPath}
\SetKwFunction{clonePath}{clonePath}
\SetKwFunction{findMostProbablePath}{findMostProbablePath}
\SetKwFunction{initializeDataStructures}{initializeDataStructures}
\SetKwData{forksArray}{forksArray}
\SetKwData{contForks}{contForks}
\SetKwData{probForks}{probForks}
\SetKwData{inactivePathIndices}{inactivePathIndices}
\SetKwData{inactiveArrayIndices}{inactiveArrayIndices}
\SetKwData{arrayReferenceCount}{arrayReferenceCount}
\SetKwData{pathIndexToArrayIndex}{pathIndexToArrayIndex}
\SetKwData{arrayPointerP}{arrayPointer\_P}
\SetKwData{arrayPointerC}{arrayPointer\_C}
\SetKwData{activePath}{activePath}
\SetKw{Continue}{continue}
\SetKw{New}{new}
\SetKw{Push}{push}
\SetKw{Pop}{pop}
\SetKw{False}{false}
\SetKw{True}{true}
\SetKw{myand}{and}
\SetKw{myor}{or}
\DontPrintSemicolon
\SetAlgoSkip{-5pt}
\newcommand{\alghfill}{\hfill\hfill\hfill\hfill\hfill\hfill\hfill\hfill\hfill\hfill\hfill\hfill\hfill\hfill\hfill\hfill\hfill\hfill\hfill\hfill\hfill\hfill\hfill\hfill}
\begin{abstract}
We describe a successive-cancellation \emph{list} decoder for polar codes, which is a generalization of the classic successive-cancellation decoder of Ar{\i}kan. In the proposed list decoder, up to $L$ decoding paths are considered concurrently at each decoding stage. Then, a single codeword is selected from the list as output. If the most likely codeword is selected, simulation results show that the resulting performance is very close to that of a maximum-likelihood decoder, even for moderate values of $L$. Alternatively, if a ``genie'' is allowed to pick the codeword from the list, the results are comparable to the current state of the art LDPC codes. Luckily, implementing such a helpful genie is easy. 

Our list decoder doubles the number of decoding paths at each decoding step, and then uses a pruning procedure to discard all but the $L$ ``best'' paths. 
Nevertheless, a straightforward implementation still requires $\Omega(L \cdot n^2)$  time, which is in stark contrast with the $O(n \log n)$ complexity of the original successive-cancellation decoder. We utilize the structure of polar codes to overcome this problem. Specifically, we devise an efficient, numerically stable, implementation taking only $O(L \cdot n \log n)$ time and $O(L \cdot n)$ space. 
\end{abstract}

\addtolength{\textfloatsep}{-7mm}
\section{Introduction}
Polar codes, recently discovered by \Arikan\ \cite{Arikan:09p}, are a major breakthrough in coding theory. They are the first and currently only family of codes known to have an explicit construction (no ensemble to pick from) and efficient encoding and decoding algorithms, while also being capacity achieving over binary input symmetric memoryless channels. Their probability of error is known to approach $O(2^{-\sqrt{n}})$ \cite{ArikanTelatar:09c}, with generalizations giving even better asymptotic results \cite{KSU:10p}.

Of course, ``capacity achieving'' is an asymptotic property, and the main sticking point of polar codes to date is that their performance at short to moderate block lengths is disappointing. As we ponder why, we identify two possible culprits: either the codes themselves are inherently weak at these lengths, or the successive cancellation (SC) decoder employed to decode them is significantly degraded with respect to Maximum Likelihood (ML) decoding performance. More so, the two possible culprits are complementary, and so both may occur.

In this paper we show an improvement to the SC decoder, namely, a successive cancellation list (SCL) decoder. Our list decoder has a corresponding \emph{list size} $L$, and setting $L=1$ results in the classic SC decoder. It should be noted that the word ``list'' was chosen as part of the name of our decoder in order to highlight a key concept relating to the inner working of it. However, when our algorithm finishes, it returns a \emph{single} codeword.

\begin{figure}
\begin{center}
\includegraphics[scale=0.50]{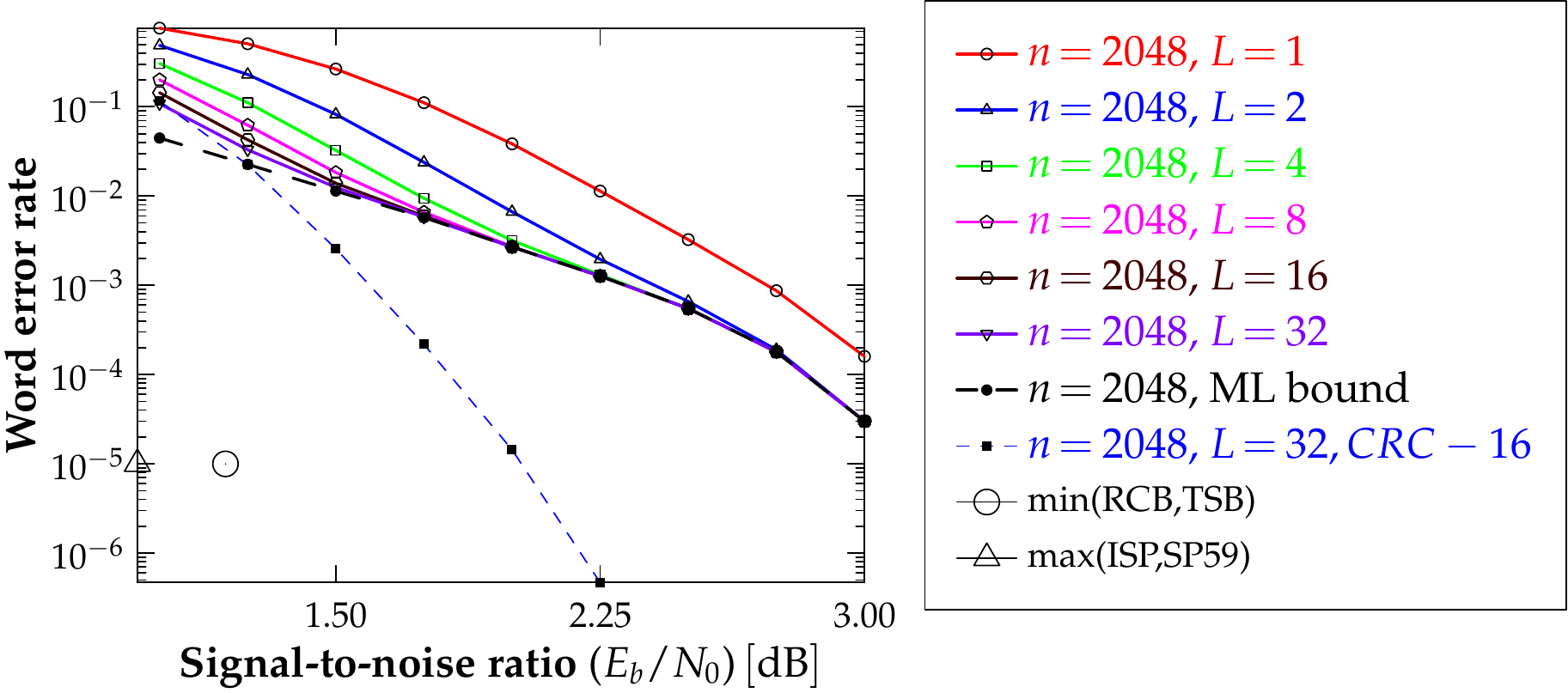}

%
\caption{Word error rate of a length $n=2048$ rate $1/2$ polar code optimized for SNR=$2$ dB under various list sizes. Code construction was carried out via the method proposed in \cite{TalVardy:11a}. The two dots represent upper and lower bounds \cite{WiechmanSason:08p} on the SNR needed to reach a word error rate of $10^{-5}$.}
\label{fig:WER}
\end{center}
\end{figure}

The solid lines in Figure~\ref{fig:WER} corresponds to choosing the most likely codeword from the list as the decoder output. As can be seen, this choice of the most likely codeword results in a large range in which our algorithm has performance very close to that of the ML decoder, even for moderate values of $L$. Thus, the sub-optimality of the SC decoder indeed does plays a role in the disappointing performance of polar codes. 

\begin{figure}
\begin{center}
\includegraphics[scale=0.45]{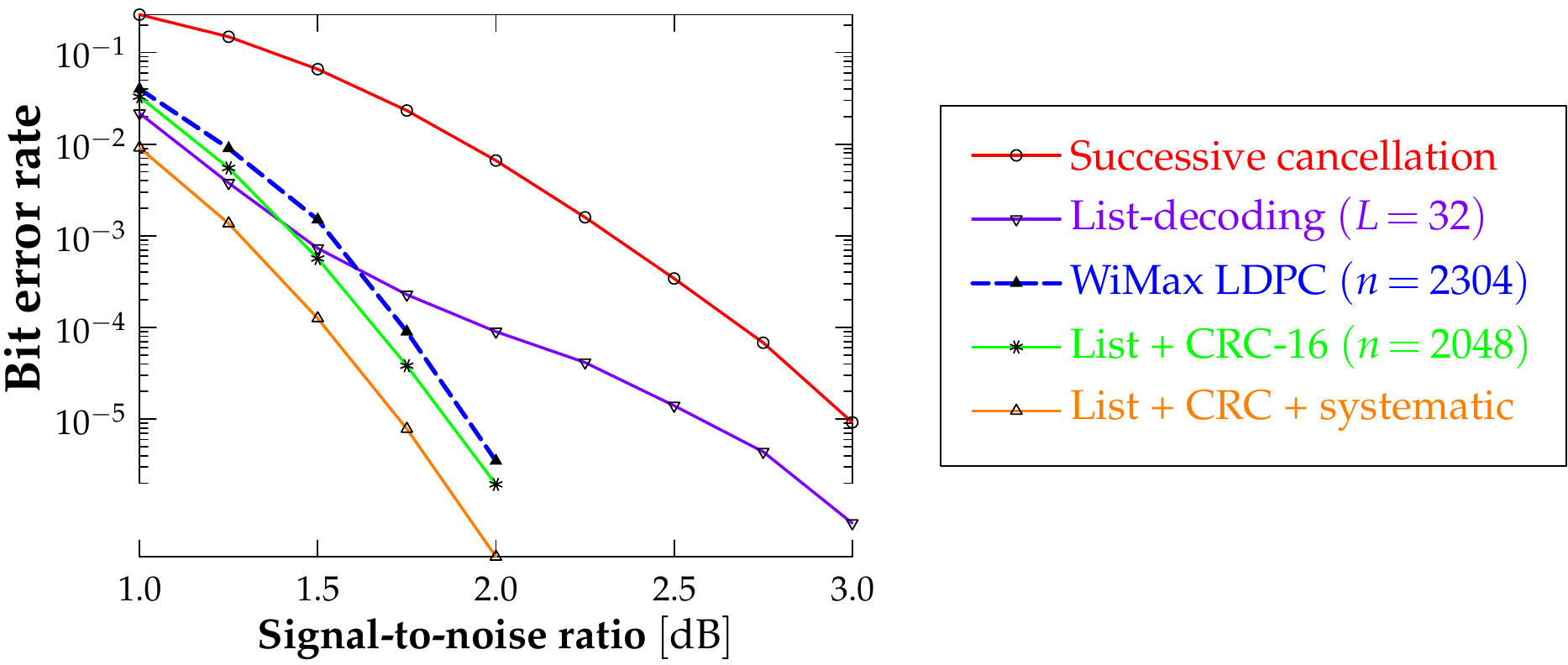}
\caption{Comparison of our polar coding and decoding schemes to an implementation of the WiMax standard take from \cite{Turbobest:00u}. All codes are rate $1/2$. The length of the polar code is $2048$ while the length of the WiMax code is $2304$. The list size used was $L=32$. The CRC used was $16$ bits long.}
\label{fig:WiMax}
\end{center}
\end{figure}

Even with the above improvement, the performance of polar-codes falls short. Thus, we conclude that polar-codes themselves are weak. Luckily, we can do better. Suppose that instead of picking the most likely codeword from the list, a ``genie'' would aid us by telling us what codeword in the list was the transmitted codeword (if the transmitted codeword was indeed present in the list). Luckily, implementing such a genie turns out to be simple, and entails a slight modification of the polar code. With this modification, the performance of polar codes is comparable to state of the art LDPC codes, as can be seen in Figure~\ref{fig:WiMax}.

In fairness, we refer to Figure~\ref{fig:PPV} and note that there are LDPC codes of length $2048$ and rate $1/2$ with better performance than our polar codes. However, to the best of our knowledge, for length $1024$ and rate $1/2$ it seems that our implementation is slightly better than previously known codes when considering a target error-probability of $10^{-4}$.

\begin{figure}
\begin{center}
\includegraphics[scale=0.30]{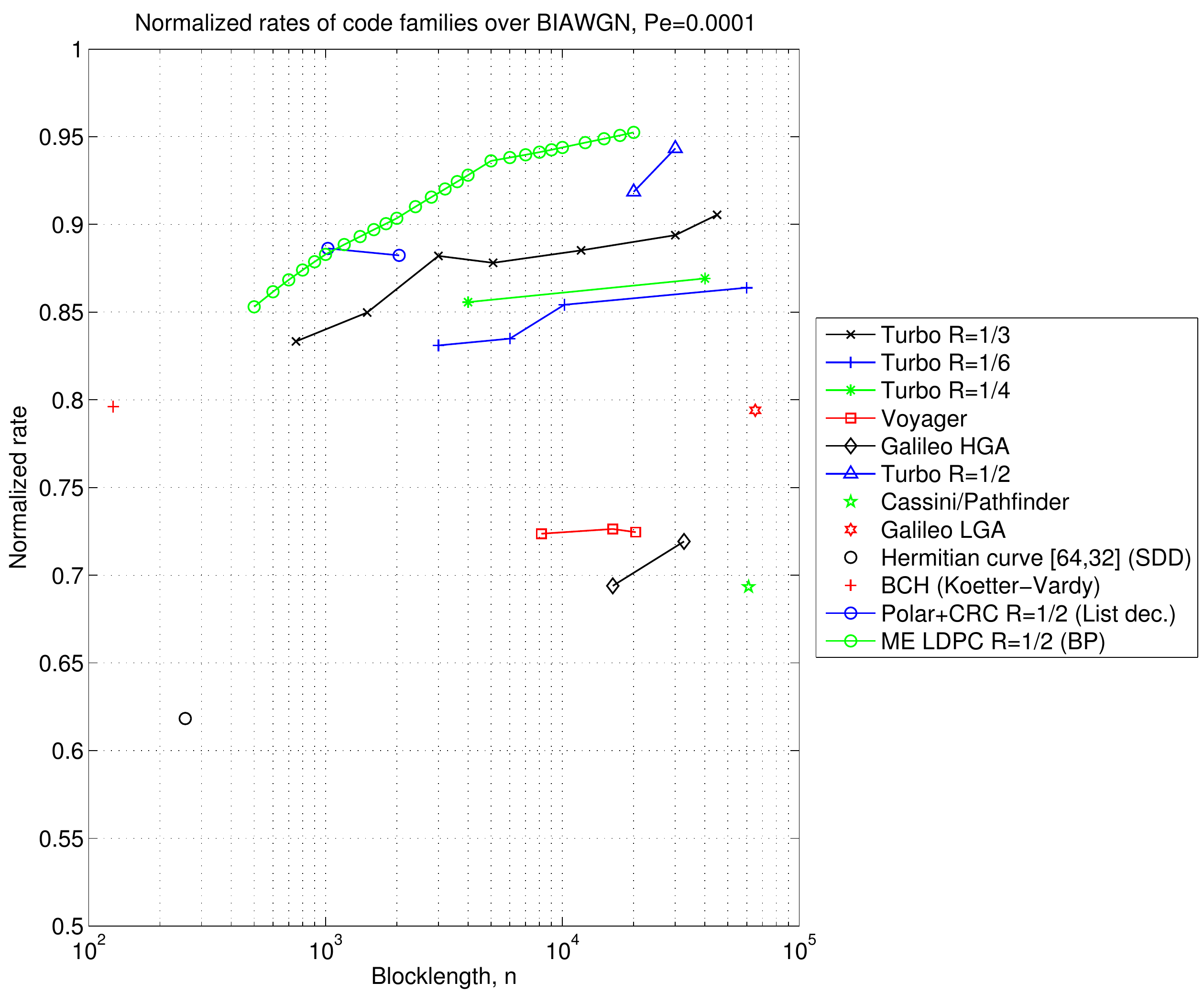}
\caption{Comparison of normalized rate \cite{PPV:10p} for a wide class of codes. The target word error rate is $10^{-4}$. The plot is courtesy of Dr.\ Yury Polyanskiy.}
\label{fig:PPV}
\end{center}
\end{figure}

The structure of this paper is as follows. In Section~\ref{sec:SC}, we present \Arikan's SC decoder in a notation that will be useful to us later on. In Section~\ref{sec:SCnspace}, we show how the space complexity of the SC decoder can be brought down from $O(n \log n)$ to $O(n)$. This observation will later help us in Section~\ref{sec:SCList}, where we presents our successive cancellation list decoder with time complexity $O(L\cdot n \log n)$. Section~\ref{sec:listCRC} introduces a modification of polar codes which, when decoded with the SCL decoder, results in a significant improvement in terms of error rate.

This paper contains a fair amount of algorithmic detail. Thus, on a first read, we advise the reader to skip to Section~\ref{sec:SCList} and read the first three paragraphs. Doing so will give a high-level understanding of the decoding method proposed and also show why a naive implementation is too costly. Then, we advise the reader to skim Section~\ref{sec:listCRC} where the ``list picking genie'' is explained.

\section{Formalization of the Successive Cancellation Decoder}
\label{sec:SC}

The Successive Cancellation (SC) decoder is due to \Arikan\ \cite{Arikan:09p}. In this section, we recast it using our notation, for future reference. 

Let the polar code under consideration have length $n=2^m$ and dimension $k$. Thus, the number of frozen bits is $n-k$. We denote by $\sfbu = (\sfu_i)_{i=0}^{n-1} = \sfu_0^{n-1}$ the information bits vector (including the frozen bits), and by $\sfbc=\sfc_{0}^{n-1}$ the corresponding codeword, which is sent over a binary-input channel $W:\calX \to \calY$, where $\calX=\{0,1\}$. At the other end of the channel, we get the received word $\sfby = \sfy_{0}^{n-1}$. A decoding algorithm is then applied to $\sfby$, resulting in a decoded codeword $\hat{\sfbc}$ having corresponding information bits $\hat{\sfbu}$.

\subsection{An outline of Successive Cancellation}
A high-level description of the SC decoding algorithm is given in Algorithm~\ref{alg:highLevel}. In words, at each \emph{phase} $\phaseVar$ of the algorithm, we must first calculate the pair of probabilities $W_m^{(\phaseVar)}(\sfby_0^{n-1}, \hat{\sfbu}_0^{\phaseVar-1}|0)$ and $W_m^{(\phaseVar)}(\sfby_0^{n-1}, \hat{\sfbu}_0^{\phaseVar-1}|1)$, defined shortly. Then, we must make a decision as to the value of $\hat{\sfu}_\phaseVar$ according to the pair of probabilities.

\begin{algorithm}
\SetInd{0.49em}{0.49em}
\caption{A high-level description of the SC decoder}
\label{alg:highLevel}
\KwIn{the received vector $\sfby$}
\KwOut{a decoded codeword $\hat{\sfbc}$}
\BlankLine
\For{$\phaseVar = 0,1,\ldots, n-1$}
{
  calculate $W_m^{(\phaseVar)}(\sfby_0^{n-1}, \hat{\sfbu}_0^{\phaseVar-1}|0)$ and $W_m^{(\phaseVar)}(\sfby_0^{n-1}, \hat{\sfbu}_0^{\phaseVar-1}|1)$\; \label{alg:highLevel:calcW}
  \eIf{$\sfu_\phaseVar$ is frozen}
  {
    set $\hat{\sfu}_\phaseVar$ to the frozen value of $\sfu_\phaseVar$ \;
  }
  {
    \eIf{$W_m^{(\phaseVar)}(\sfby_0^{n-1}, \hat{\sfbu}_0^{\phaseVar-1}|0) > W_m^{(\phaseVar)}(\sfby_0^{n-1}, \hat{\sfbu}_0^{\phaseVar-1}|1)$}
    {
      set $\hat{\sfu}_\phaseVar \gets 0$\;
    }
    {
      set $\hat{\sfu}_\phaseVar \gets 1$ \;
    }
  }
}
\Return the codeword $\hat{\sfbc}$ corresponding to $\hat{\sfbu}$
\end{algorithm}

We now show how the above probabilities are calculated. For \emph{layer} $0 \leq \layerVar \leq m$, denote hereafter
\begin{equation}
\layerCapitalVar = 2^\layerVar \; .
\end{equation}
Recall \cite{Arikan:09p} that for 
\begin{equation}
0 \leq \phaseVar < \layerCapitalVar \; ,
\end{equation}
\emph{bit channel} $W_{\layerVar}^{(\phaseVar)}$  is a binary input channel with output alphabet $\calY^{\layerCapitalVar} \times \calX^{\phaseVar}$, the conditional probability of which we generically denote as 
\begin{equation}
W_\layerVar^{(\phaseVar)}(\bfgenericReceived_0^{\layerCapitalVar-1}, \bfu_{0}^{\phaseVar-1}|u_\phaseVar) \; .
\end{equation}
In our context, $\bfgenericReceived_0^{\layerCapitalVar-1}$ is always a contiguous subvector of received vector $\sfby$. Next, for $1 \leq \layerVar \leq m$, recall the recursive definition of a bit channel \cite[Equations (22) and (23)]{Arikan:09p} : let $0 \leq 2\psi < \layerCapitalVar$, then
\begin{multline}
\label{eq:ArikanFirstRecursive}
\overbrace{W_{\layerVar}^{(2\psi)}(\bfgenericReceived_0^{\layerCapitalVar-1},\bfu_0^{2\psi-1}|u_{2\psi})}^{\mbox{branch $\beta$}} \\
= \sum_{u_{2\psi+1}} \frac{1}{2}
\underbrace{W_{\layerVar-1}^{(\psi)}(\bfgenericReceived_0^{\layerCapitalVar/2-1}, \bfu_{0,\subscripteven}^{2\psi-1} \xor \bfu_{0,\subscriptodd}^{2\psi-1}|u_{2\psi} \xor u_{2\psi+1})}_{\mbox{branch $2\beta$}}\\
\cdot \underbrace{W_{\layerVar-1}^{(\psi)}(\bfgenericReceived_{\layerCapitalVar/2}^{\layerCapitalVar-1}, \bfu_{0,\subscriptodd}^{2\psi-1}|u_{2\psi+1})}_{\mbox{branch $2\beta+1$}}
\end{multline}
and
\begin{multline}
\label{eq:ArikanSecondRecursive}
\overbrace{W_{\layerVar}^{(2\psi+1)}(\bfgenericReceived_0^{\layerCapitalVar-1},\bfu_0^{2\psi}|u_{2\psi+1})}^{\mbox{branch $\beta$}} \\
= \frac{1}{2}
\underbrace{W_{\layerVar-1}^{(\psi)}(\bfgenericReceived_0^{\layerCapitalVar/2-1}, \bfu_{0,\subscripteven}^{2\psi-1} \xor \bfu_{0,\subscriptodd}^{2\psi-1}|u_{2\psi} \xor u_{2\psi+1})}_{\mbox{branch $2 \beta$}}\\
\cdot \underbrace{W_{\layerVar-1}^{(\psi)}(\bfgenericReceived_{\layerCapitalVar/2}^{\layerCapitalVar-1}, \bfu_{0,\subscriptodd}^{2\psi-1}|u_{2\psi+1})}_{\mbox{branch $2\beta+1$}}
\end{multline}
with ``stopping condition'' $W_0^{(0)}(\genericReceived|u) = W(\genericReceived|u)$.

\subsection{Detailed description}
For Algorithm~\ref{alg:highLevel} to become concrete, we must specify how the probability pair associated with $W_m^{(\phaseVar)}$ is calculated, and how the set values of $\hat{\sfbu}$, namely $\hat{\sfbu}_0^{\phaseVar-1}$, are propagated into those calculations. We now show an implementation that is straightforward, yet somewhat wasteful in terms of space.

For $\layerVar>0$ and $0 \leq \phaseVar < \layerCapitalVar$, recall the recursive definition of $W_\layerVar^{(\phaseVar)}(\bfgenericReceived_0^{\Lambda-1}, \bfu_0^{\phaseVar-1}|u_\phaseVar)$ given in either (\ref{eq:ArikanFirstRecursive}) or (\ref{eq:ArikanSecondRecursive}), depending on the parity of $\phaseVar$. For either  $\phaseVar=2\psi$ or $\phaseVar=2\psi+1$, the channel $W_{\layerVar-1}^{(\psi)}$ is evaluated with output $(\bfgenericReceived_0^{\Lambda/2-1}, \bfu_{0,\subscripteven}^{2\psi-1} \xor \bfu_{0,\subscriptodd}^{2\psi-1})$, as well as with output  $(\bfgenericReceived_{\Lambda/2}^{\Lambda-1}, \bfu_{0,\subscriptodd}^{2\psi-1})$. Since our algorithm will make use of these recursions, we need a simple way of defining which output we are referring to. We do this by specifying, apart from the layer $\layerVar$ and the phase $\phaseVar$ which define the channel, the \emph{branch} number 
\begin{equation}
\label{eq:branchNumber}
0 \leq \branchVar < 2^{m-\lambda} \; . 
\end{equation}

Since, during the run of the SC algorithm, the channel $W_m^{(\phaseVar)}$ is only evaluated with a single output, $(\sfby_0^{n-1},\hat{\sfbu}_0^{\phaseVar-1})$, we give a branch number of $\branchVar = 0$ to each such output.  Next, we proceed recursively as follows. For $\layerVar > 0$, consider a channel $W_\layerVar^{(\phaseVar)}$ with output $(\bfgenericReceived_0^{\Lambda-1}, \hat{\bfu}_0^{\phaseVar-1})$ and corresponding branch number $\branchVar$. Denote $\psi = \floor{\phaseVar/2}$. The output $(\bfgenericReceived_0^{\Lambda/2-1}, \hat{\bfu}_{0,\subscripteven}^{2\psi-1} \xor \hat{\bfu}_{0,\subscriptodd}^{2\psi-1})$ associated with $W_{\layerVar-1}^{(\psi)}$ will have a branch number of $2\branchVar$, while the output $(\bfgenericReceived_{\Lambda/2}^{\Lambda-1}, \hat{\bfu}_{0,\subscriptodd}^{2\psi-1})$ will have a branch number of $2\branchVar+1$. Finally, we mention that for the sake of brevity, we will talk about the output corresponding to \emph{branch $\beta$ of a channel}, although this is slightly inaccurate. 


We now introduce our first data structure. For each layer $0 \leq \layerVar \leq m$, we will have a \emph{probabilities array}, denoted by $P_\layerVar$, indexed by an integer $0 \leq i < 2^m$ and a bit $b \in \myset{0,1}$. For a given layer $\layerVar$, an index $i$ will correspond to a phase $0 \leq \phaseVar < \layerCapitalVar$ and branch $0 \leq \branchVar < 2^{m-\layerVar}$ using the following quotient/reminder representation.
\begin{equation}
\label{eq:phaseBranchIndexing}
i = \phaseAndBranch{\layerVar}{\phaseVar,\branchVar} = \phaseVar + 2^{\layerVar} \cdot \branchVar \; .
\end{equation}
In order to avoid repetition, we use the following shorthand
\begin{equation}
\label{eq:arrayShorthand}
\arrayPhaseAndBranch{P}{\layerVar}{\phaseVar,\branchVar} = P_\layerVar[\phaseAndBranch{\layerVar}{\phaseVar,\branchVar}] 
\; .
\end{equation}

The probabilities array data structure $P_\layerVar$ will be used as follows.  Let a layer $0 \leq \layerVar \leq m$, phase $0 \leq \phaseVar < \layerCapitalVar$, and branch $0 \leq \branchVar < 2^{m-\lambda}$ be given. Denote the output corresponding to branch $\branchVar$ of $W_\layerVar^{(\phaseVar)}$ as $(\bfgenericReceived_0^{\layerCapitalVar-1}, \hat{\bfu}_{0}^{\phaseVar-1})$. Then, ultimately, we will have for both values of $b$ that 
\begin{equation}
\label{eq:probArray}
\arrayPhaseAndBranch{P}{\layerVar}{\phaseVar,\branchVar}[b]= W_\layerVar^{(\phaseVar)}(\bfgenericReceived_0^{\layerCapitalVar-1}, \hat{\bfu}_{0}^{\phaseVar-1}|b) \; . 
\end{equation}

Analogously to defining the output corresponding to a branch $\branchVar$, we would now like define the input corresponding to a branch. As in the ``output'' case, we start at layer $m$ and continue recursively. 
Consider the channel $W_m^{(\phaseVar)}$, and let $\hat{\sfu}_\phaseVar$ be the corresponding input which Algorithm~\ref{alg:highLevel} assumes. We let this input have a branch number of $\branchVar = 0$. Next, we proceed recursively as follows. For layer $\layerVar > 0$, consider the channels $W_\layerVar^{(2 \psi )}$ and $W_\layerVar^{(2 \psi +1 )}$ having the same branch $\branchVar$ with corresponding inputs $u_{2 \psi}$ and $u_{2 \psi +1 }$, respectively. In light of (\ref{eq:ArikanSecondRecursive}), we now consider $W_{\layerVar-1}^{(\psi)}$ and define the input corresponding to branch $2 \branchVar$  as $u_{2 \psi} \xor u_{2 \psi +1 }$. Likewise, we define the input corresponding to branch $2 \branchVar+1$ as $u_{2 \psi+1}$. Note that under this recursive definition, we have that for all $0 \leq \layerVar \leq m$, $0 \leq \phaseVar < \layerCapitalVar$, and $0 \leq \branchVar < 2^{m-\lambda}$, the input corresponding to branch $\branchVar$ of $W_\layerVar^{(\phaseVar)}$ is well defined.  

The following lemma points at the natural meaning that a branch number has at layer $\layerVar=0$. It is proved using a straightforward induction.
\begin{lemm}
Let $\sfby$ and $\hat{\sfbc}$ be as in Algorithm~\ref{alg:highLevel}, the received vector and the decoded codeword. Consider layer $\layerVar=0$, and thus set $\phaseVar=0$. Next, fix a branch number $0 \leq \branchVar < 2^n$. Then, the input and output corresponding to branch $\branchVar$ of $W_0^{(0)}$ are $\sfy_\branchVar$ and $\hat{\sfc}_\branchVar$, respectively. 
\end{lemm}

We now introduce our second, and last, data structure for this section. For each layer $0 \leq \layerVar \leq m$, we will have a \emph{bit array}, denoted by $B_\layerVar$, and indexed by an integer $0 \leq i < 2^m$, as in (\ref{eq:phaseBranchIndexing}). The data structure will be used as follows.  Let layer $0 \leq \layerVar \leq m$, phase $0 \leq \phaseVar < \layerCapitalVar$, and  branch $0 \leq \branchVar < 2^{m-\lambda}$ be given. Denote the input corresponding to branch $\branchVar$ of $W_\layerVar^{(\phaseVar)}$ as $\hat{u}(\layerVar,\phaseVar, \branchVar)$. Then, ultimately, 
\begin{equation}
\label{eq:bitArray}
\arrayPhaseAndBranch{B}{\layerVar}{\phaseVar,\branchVar} = \hat{u}(\layerVar,\phaseVar, \branchVar) \; , 
\end{equation}
where we have used the same shorthand as in (\ref{eq:arrayShorthand}). Notice that the total memory consumed by our algorithm is $O(n \log n)$.

Our first implementation of the SC decoder is given as Algorithms~\ref{alg:firstImplementation_main}--\ref{alg:firstImplementation_updateB}. The main loop is given in Algorithm~\ref{alg:firstImplementation_main}, and follows the high-level description given in Algorithm~\ref{alg:highLevel}. Note that the elements of the probabilities arrays $P_\layerVar$ and bit array $B_\layerVar$ start-out uninitialized, and become initialized as the algorithm runs its course. The code to initialize the array values is given in Algorithms~\ref{alg:firstImplementation_calcP} and \ref{alg:firstImplementation_updateB}.

\begin{algorithm}
\caption{First implementation of SC decoder}
\label{alg:firstImplementation_main}
\KwIn{the received vector $\sfby$}
\KwOut{a decoded codeword $\hat{\sfbc}$}
\BlankLine
\For(\tcp*[h]{Initialization}){$\branchVar = 0,1,\ldots, n-1$}
{
  $\arrayPhaseAndBranch{P}{0}{0,\branchVar}[0] \gets W(\sfy_\branchVar|0)$, $\arrayPhaseAndBranch{P}{0}{0,\branchVar}[1] \gets W(\sfy_\branchVar|1)$\;
}
\For(\tcp*[h]{Main loop}){$\phaseVar = 0,1,\ldots, n-1$}
{ 
  $\recursivelyCalcP(m, \phaseVar)$\; \label{alg:firstImplementation_main:recursivelyCalcP}
  \eIf{$\sfu_\phaseVar$ is frozen}
  {
    set $\arrayPhaseAndBranch{B}{m}{\phaseVar,0}$ to the frozen value of $\sfu_\phaseVar$ \;
  }
  {
    \eIf{$\arrayPhaseAndBranch{P}{m}{\phaseVar,0}[0] > \arrayPhaseAndBranch{P}{m}{\phaseVar,0}[1]$}
    { \label{alg:firstImplementation_main:unfrozenStart}
      set $\arrayPhaseAndBranch{B}{m}{\phaseVar,0} \gets 0$ \;
    }
    {
      set $\arrayPhaseAndBranch{B}{m}{\phaseVar,0} \gets 1$ \; \label{alg:firstImplementation_main:unfrozenEnd}
    }
  }
  \If{$\phaseVar \mod 2 = 1$}
  {
    $\recursivelyUpdateB(m, \phaseVar)$\;
  }
}
\Return the decoded codeword: $\hat{\sfbc}=\left(\arrayPhaseAndBranch{B}{0}{0,\branchVar}\right)_{\branchVar=0}^{n-1}$\;
\end{algorithm}
\begin{algorithm}
\SetInd{0.52em}{0.52em}
\caption{recursivelyCalcP$(\layerVar, \phaseVar)$\hfill\emph{implementation I}}
\label{alg:firstImplementation_calcP}
\KwIn{layer $\layerVar$ and phase $\phaseVar$}
\BlankLine
\lIf{$\layerVar = 0$}{\Return\tcp*[h]{Stopping condition}\;}
set $\psi \gets \floor{\phaseVar/2}$\;
\tcp{Recurse first, if needed}
\lIf{$\phaseVar \mod 2 = 0$}
{
  $\recursivelyCalcP(\layerVar-1, \psi)$\;
}

\For(\tcp*[h]{calculation}){$\branchVar=0,1,\ldots,2^{m-\layerVar}-1$}
{
  \eIf(\tcp*[f]{apply Equation (\ref{eq:ArikanFirstRecursive})}){$\phaseVar \mod 2 = 0$}
  {
    \For{$u' \in \{0,1\}$}
    {
      \mbox{$\arrayPhaseAndBranch{P}{\layerVar}{\phaseVar,\branchVar}[u']\gets$} \mbox{$\sum_{u''} \frac{1}{2}
      \arrayPhaseAndBranch{P}{\layerVar-1}{\psi,2\branchVar}[u' \xor u''] 
      \cdot {}$}\;
      \alghfill $\arrayPhaseAndBranch{P}{\layerVar-1}{\psi,2\branchVar+1}[u'']$\;
    }
  }
  (\tcp*[f]{apply Equation (\ref{eq:ArikanSecondRecursive})})
  {
    set $u' \gets \arrayPhaseAndBranch{B}{\layerVar}{\phaseVar-1,\branchVar}$\;
    \For{$u'' \in \{0,1\}$}
    {
      \mbox{$\arrayPhaseAndBranch{P}{\layerVar}{\phaseVar,\branchVar}[u''] \gets$} \mbox{$\frac{1}{2}
      \arrayPhaseAndBranch{P}{\layerVar-1}{\psi,2\branchVar}[u' \xor u''] 
      \cdot {}$}\;
       \alghfill $\arrayPhaseAndBranch{P}{\layerVar-1}{\psi,2\branchVar+1}[u'']$\;
    }
  }
}
\end{algorithm}

\begin{algorithm}
\caption{recursivelyUpdateB$(\layerVar,\phaseVar)$\hfill\emph{implementation I}}
\label{alg:firstImplementation_updateB}
\Require{$\phaseVar$ is odd}
\BlankLine
set $\psi \gets \floor{\phaseVar/2}$\;
\For{$\branchVar=0,1,\ldots,2^{m-\layerVar}-1$}
{
$\arrayPhaseAndBranch{B}{\layerVar-1}{\psi,2\branchVar} \gets \arrayPhaseAndBranch{B}{\layerVar}{\phaseVar-1,\branchVar} \xor \arrayPhaseAndBranch{B}{\layerVar}{\phaseVar,\branchVar}$\;
$\arrayPhaseAndBranch{B}{\layerVar-1}{\psi,2\branchVar+1} \gets \arrayPhaseAndBranch{B}{\layerVar}{\phaseVar,\branchVar}$ \;
}
\If{$\psi \mod 2 = 1$}
{
$\recursivelyUpdateB(\layerVar-1, \psi)$\;
}
\end{algorithm}
\begin{lemm}
Algorithms~\ref{alg:firstImplementation_main}--\ref{alg:firstImplementation_updateB} are a valid implementation of the SC decoder.
\end{lemm}
\begin{proof}
We first note that in addition to proving the claim explicitly stated in the lemma, we must also prove an implicit claim. Namely, we must prove that the actions taken by the algorithm are well defined. Specifically, we must prove that when an array element is read from,  it was already written to (it is initialized).

Both the implicit and explicit claims are easily derived from the following observation. For a given $0 \leq \phaseVar < n$, consider iteration $\phaseVar$ of the main loop in Algorithm~\ref{alg:firstImplementation_main}. Fix a layer $0 \leq \layerVar \leq m$, and a branch $0 \leq \branchVar < 2^{m-\layerVar}$. If we suspend the run of the algorithm just after the iteration ends, then (\ref{eq:probArray}) holds with $\phaseVar'$ instead of $\phaseVar$, for all 
\[
0 \leq \phaseVar' \leq \floorr{\frac{\phaseVar}{2^{m-\layerVar}}} \; .  
\]
Similarly, (\ref{eq:bitArray}) holds with $\phaseVar'$ instead of $\phaseVar$, for all 
\[
0 \leq \phaseVar' < \floorr{\frac{\phaseVar+1}{2^{m-\layerVar}}} \; .
\]

The above observation is proved by induction on $\phaseVar$.
\end{proof}

\section{Space-Efficient Successive Cancellation Decoding}
\label{sec:SCnspace}
The running time of the SC decoder is $O(n \log n)$, and our implementation is no exception. As we have previously noted, the space complexity of our algorithm is $O( n \log n)$ as well. However, we will now show how to bring the space complexity down to $O(n)$. The observation that one can reduce the space complexity to $O(n)$ was noted, in the context of VLSI design, in \cite{LTVG:10a}.

As a first step towards this end, consider the probability pair array $P_m$. By examining the main loop in Algorithm~\ref{alg:firstImplementation_main}, we quickly see that if we are currently at phase $\phaseVar$, then we will never again make use of $\arrayPhaseAndBranch{P}{m}{\phaseVar',0}$ for all $\phaseVar' < \phaseVar$. On the other hand, we see that $\arrayPhaseAndBranch{P}{m}{\phaseVar'',0}$ is uninitialized for all $\phaseVar'' > \phaseVar$. Thus, instead of reading and writing to $\arrayPhaseAndBranch{P}{m}{\phaseVar,0}$, we can essentially disregard the phase information, and use only the first element $P_m[0]$ of the array, discarding all the rest. By the recursive nature of polar codes, this observation --- disregarding the phase information --- can be exploited for a general layer $\layerVar$ as well. Specifically, for all $0 \leq \layerVar \leq m$, let us now define the number of elements in $P_\layerVar$ to be $2^{m-\layerVar}$. Accordingly,
\begin{equation}
\label{eq:PReplace}
\mbox{$\arrayPhaseAndBranch{P}{\layerVar}{\phaseVar,\branchVar}$ is replaced by $P_{\layerVar}[\branchVar]$} \; .
\end{equation}

Note that the total space needed to hold the $P$ arrays has gone down from $O(n \log n)$ to $O(n)$. We would now like to do the same for the $B$ arrays. However, as things are currently stated, we can not disregard the phase, as can be seen for example in line 3 of Algorithm~\ref{alg:firstImplementation_updateB}. The solution is a simple renaming. As a first step, let us define for each $0 \leq \layerVar \leq m$ an array $C_\layerVar$ consisting of bit pairs and having length $n/2$. Next, let a generic reference of the form $\arrayPhaseAndBranch{B}{\layerVar}{\phaseVar,\branchVar}$ be replaced by $C_\layerVar[\psi+\branchVar \cdot 2^{\layerVar-1}][\phaseVar \mod 2]$, where $\psi = \floor{\phaseVar/2}$.
Note that we have done nothing more than rename the elements of $B_\layerVar$ as elements of $C_\layerVar$. However, we now see that as before we can disregard the value of $\psi$ and take note only of the parity of $\phaseVar$. So, let us make one more substitution: replace every instance of $C_\layerVar[\psi+\branchVar \cdot 2^{\layerVar-1}][\phaseVar \mod 2]$ by $C_\layerVar[\branchVar][\phaseVar \mod 2]$, and resize each array $C_\layerVar$ to have $2^{m-\layerVar}$ bit pairs. To sum up,
\begin{equation}
\label{eqn:BCReplace}
\mbox{$\arrayPhaseAndBranch{B}{\layerVar}{\phaseVar,\branchVar}$ is replaced by $C_\layerVar[\branchVar][\phaseVar \mod 2]$} \; .
\end{equation}

The alert reader will notice that a further reduction in space is possible: for $\layerVar=0$ we will always have that $\phaseVar=0$, and thus the parity of $\phaseVar$ is always even. However, this reduction does not affect the asymptotic space complexity which is now indeed down to $O(n)$. The revised algorithm is given as Algorithms~\ref{alg:improvedSC_main}--\ref{alg:improvedSC_updateC}.

\label{app:improvedSC}
\begin{algorithm}
\KwIn{the received vector $\sfby$}
\KwOut{a decoded codeword $\hat{\sfbc}$}
\BlankLine
\caption{Space efficient SC decoder, main loop}
\label{alg:improvedSC_main}
\For(\tcp*[h]{Initialization}){$\branchVar = 0,1,\ldots, n-1$}{
  set $\arrayAndBranch{P}{0}{\branchVar}[0] \gets W(\sfy_\branchVar|0)$, $\arrayAndBranch{P}{0}{\branchVar}[1] \gets W(\sfy_\branchVar|1)$
  }
\For(\tcp*[h]{Main loop}){$\phaseVar = 0,1,\ldots, n-1$}
{
  $\recursivelyCalcP(m, \phaseVar)$\;
  \eIf{$\sfu_\phaseVar$ is frozen}
  {
    set $\arrayAndBranch{C}{m}{0}[\phaseVar \mod 2]$ to the frozen value of $\sfu_\phaseVar$ \; \label{alg:improvedSC_main:frozenSet}
  }
  { 
    \eIf{$\arrayAndBranch{P}{m}{0}[0] > \arrayAndBranch{P}{m}{0}[1]$} 
    { \label{alg:improvedSC_main:unfrozenStart}
      set $\arrayAndBranch{C}{m}{0}[\phaseVar \mod 2] \gets 0$
    }
    { 
      set $\arrayAndBranch{C}{m}{0}[\phaseVar \mod 2] \gets 1$ \; \label{alg:improvedSC_main:unfrozenEnd}
    }
  }
  \If{$\phaseVar \mod 2 = 1$}
  {
    $\recursivelyUpdateC(m, \phaseVar)$
  }
}
\Return the decoded codeword: $\hat{\sfbc}=\left(\arrayAndBranch{C}{0}{\branchVar}[0]\right)_{\branchVar=0}^{n-1}$
\end{algorithm}

\begin{algorithm}
\caption{recursivelyCalcP$(\layerVar, \phaseVar)$\hfill \emph{space-efficient}}
\label{alg:improvedSC_calcP}
\KwIn{layer $\layerVar$ and phase $\phaseVar$}
\BlankLine
\lIf{$\layerVar = 0$}{\Return \tcp*[h]{Stopping condition}\; }
set $\psi \gets \floor{\phaseVar/2}$\;
\tcp{Recurse first, if needed}
\lIf{$\phaseVar \mod 2 = 0$}{recursivelyCalcP$(\layerVar-1, \psi)$}\;
\tcp{Perform the calculation}
\For{$\branchVar=0,1,\ldots,2^{m-\layerVar}-1$}
{
  \eIf(\tcp*[f]{apply Equation (\ref{eq:ArikanFirstRecursive})}){$\phaseVar \mod 2 = 0$}
  {
    \For{$u' \in \{0,1\}$}
    {
      $\arrayAndBranch{P}{\layerVar}{\branchVar}[u'] \gets \sum_{u''} \frac{1}{2}
       \arrayAndBranch{P}{\layerVar-1}{2\branchVar}[u' \xor u''] 
       \cdot \arrayAndBranch{P}{\layerVar-1}{2\branchVar+1}[u'']
      $\;
    }
  }
  (\tcp*[f]{apply Equation (\ref{eq:ArikanSecondRecursive})})
  { 
    set $u' \gets \arrayAndBranch{C}{\layerVar}{\branchVar}[0]$\;
    \For{$u'' \in \{0,1\}$}
    {
      $\arrayAndBranch{P}{\layerVar}{\branchVar}[u''] \gets \frac{1}{2}
      \arrayAndBranch{P}{\layerVar-1}{2\branchVar}[u' \xor u''] 
      \cdot \arrayAndBranch{P}{\layerVar-1}{2\branchVar+1}[u'']
      $\;
    }
  }
}
\end{algorithm}

\begin{algorithm}
\caption{recursivelyUpdateC$(\layerVar,\phaseVar)$\hfill\emph{space-efficient}}
\label{alg:improvedSC_updateC}
\KwIn{layer $\layerVar$ and phase $\phaseVar$}
\Require{$\phaseVar$ is odd}
\BlankLine
set $\psi \gets \floor{\phaseVar/2}$\;
\For{$\branchVar=0,1,\ldots,2^{m-\layerVar}-1$}
{
$\arrayAndBranch{C}{\layerVar-1}{2\branchVar}[\psi \mod 2] \gets \arrayAndBranch{C}{\layerVar}{\branchVar}[0] \xor \arrayAndBranch{C}{\layerVar}{\branchVar}[1]$\;
$\arrayAndBranch{C}{\layerVar-1}{2\branchVar+1}[\psi \mod 2] \gets \arrayAndBranch{C}{\layerVar}{\branchVar}[1]$ \;
}
\If{$\psi \mod 2 = 1$}
{
  recursivelyUpdateC$(\layerVar-1, \psi)$\;
}
\end{algorithm}

We end this subsection by mentioning that although we were concerned here with reducing the \emph{space} complexity of our SC decoder, the observations made with this goal in mind will be of great use in analyzing the \emph{time} complexity of our list decoder.  

\section{Successive Cancellation List Decoder}
\label{sec:SCList}
In this section we introduce and define our algorithm, the successive cancellation list (SCL) decoder. Our list decoder has a parameter $L$, called the \emph{list size}. Generally speaking, larger values of $L$ mean lower error rates but longer running times. We note at this point that successive cancellation list decoding is not a new idea: it was applied in \cite{DumerShabunov:06p} to Reed-Muller codes\footnote{In a somewhat different version of successive cancellation than that of \Arikan's, at least in exposition.}.

Recall the main loop of an SC decoder, where at each phase we must decide on the value of $\hat{\sfu}_\phaseVar$. In an SCL decoder, instead of deciding to set the value of an unfrozen $\hat{\sfu}_\phaseVar$ to either a $0$ or a $1$, we inspect both options. Namely, when decoding a non-frozen bit, we split the decoding path into two paths (see Figure~\ref{fig:listDecodingTree}). Since each split doubles the number of paths to be examined, we must prune them, and the maximum number of paths allowed is the specified list size, $L$. Naturally, we would like to keep the ``best'' paths at each stage, and thus require a pruning criterion. Our pruning criterion will be to keep the most likely paths.

\begin{figure}
\begin{center}
\includegraphics[scale=0.51]{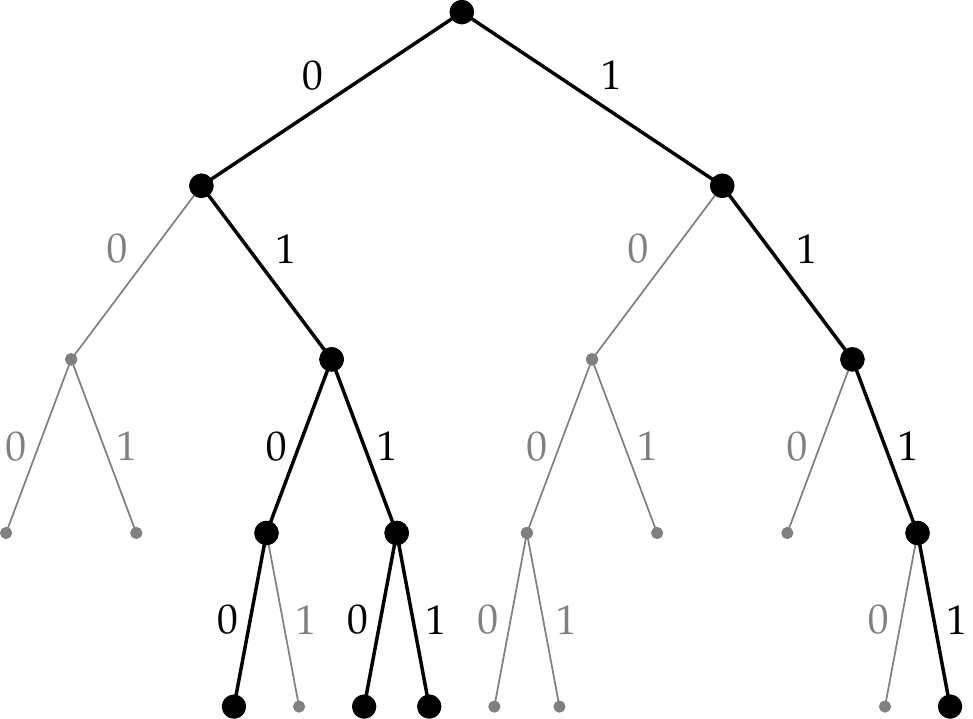}
\caption{Decoding paths of unfrozen bits for $L=4$: each level has at most $4$ nodes with paths that continue downward. Discontinued paths are colored gray.}
\label{fig:listDecodingTree}
\end{center}
\end{figure}

Consider the following outline for a naive implementation of an SCL decoder. Each time a decoding path is split into two forks, the data structures used by the ``parent'' path are duplicated, with one copy given to the first fork and the other to the second. Since the number of splits is $\Omega(L \cdot n)$, and since the size of the data structures used by each path is $\Omega(n)$, the copying operation alone would take time $\Omega(L \cdot n^2)$. This running time is clearly impractical for all but the shortest of codes. However, all known (to us) implementations of successive cancellation list decoding have complexity at least $\Omega(L \cdot n^2)$. Our main contribution in this section is the following: we show how to implement SCL decoding with time complexity $O(L \cdot n \log n)$ instead of $\Omega(L \cdot n^2)$. 

The key observation is as follows. Consider the $P$ arrays of the last section, and recall that the size of $P_\layerVar$ is proportional to $2^{m-\layerVar}$. Thus, the cost of copying $P_\layerVar$ grows exponentially small with $\layerVar$. On the other hand, looking at the main loop of Algorithm~\ref{alg:improvedSC_main} and unwinding the recursion, we see that $P_\layerVar$ is accessed only every $2^{m-\layerVar}$ incrementations of $\phaseVar$. Put another way, the bigger $P_\layerVar$ is, the less frequently it is accessed. The same observation applies to the $C$ arrays. This observation suggest the use of a ``lazy-copy''. Namely, at each given stage, the same array may be \emph{flagged as belonging to more than one decoding path}. However, when a given decoding path needs access to an array it is sharing with another path, a copy is made.

\subsection{Low-level functions}   
We now discuss the low-level functions and data structures by which the ``lazy-copy'' methodology is realized. We note in advance that since our aim was to keep the exposition as simple as possible, we have avoided some obvious optimizations. The following data structures are defined and initialized in Algorithm~\ref{alg:initializeDataStructures}.

\begin{algorithm}
\caption{initializeDataStructures$()$}
\label{alg:initializeDataStructures}
\inactivePathIndices $\gets$ \New  stack with capacity $L$\;
\activePath $\gets$ \New boolean array of size $L$\; \label{alg:initializeDataStructures:activePathDef}
\arrayPointerP $\gets$ \New 2-D array of size $(m+1) \times L$, the elements of which are array pointers\;
\arrayPointerC $\gets$ \New 2-D array of size $(m+1) \times L$, the elements of which are array pointers\;
\pathIndexToArrayIndex $\gets$ \New 2-D array of size $(m+1) \times L$\;
\inactiveArrayIndices $\gets$ \New  array of size $m+1$, the elements of which are stacks with capacity $L$\;
\arrayReferenceCount $\gets$ \New 2-D array of size $(m+1) \times L$\; \label{alg:initializeDataStructures:arrayReferenceCountDef}
\tcp{Initialization of data structures}
\For{$\layerVar=0,1,\ldots,m$}
{
  \For{$s=0,1,\ldots,L-1$}
  {
    $\arrayPointerP[\layerVar][s]$ $\gets$ \New array of float pairs of size $2^{m-\layerVar}$\;
    $\arrayPointerC[\layerVar][s]$ $\gets$ \New array of bit pairs of size $2^{m-\layerVar}$\;
    $\arrayReferenceCount[\layerVar][s]$ $\gets$ $0$\;
    $\Push(\inactiveArrayIndices[\layerVar], s)$\;
  }
}

\For{$\ell=0,1,\ldots,L-1$}
{
  $\activePath[\ell]$ $\gets$ \False\;
  $\Push(\inactivePathIndices, \ell)$\;
}
\end{algorithm}

Each path will have an index $\ell$, where $0 \leq \ell < L$. At first, only one path will be active. As the algorithm runs its course, paths will change states between ``active'' and ``inactive''. The $\inactivePathIndices$ stack \cite[Section 10.1]{CLRS:01b} will hold the indices of the inactive paths. We assume the ``array'' implementation of a stack, in which both ``push'' and ``pop'' operations take $O(1)$ time and a stack of capacity $L$ takes $O(L)$ space. The $\activePath$ array is a boolean array such that $\activePath[\ell]$ is true iff path $\ell$ is active. Note that, essentially, both $\inactivePathIndices$ and $\activePath$ store the same information. The utility of this redundancy will be made clear shortly.

For every layer $\layerVar$, we will have a ``bank'' of $L$ probability-pair arrays for use by the active paths. At any given moment, some of these arrays might be used by several paths, while others might not be used by any path. Each such array is pointed to by an element of $\arrayPointerP$. Likewise, we will have a bank of bit-pair arrays, pointed to by elements of $\arrayPointerC$.
 
The $\pathIndexToArrayIndex$ array is used as follows. For a given layer $\layerVar$ and path index $\ell$, the probability-pair array and bit-pair array corresponding to layer $\layerVar$ of path $\ell$ are pointed to by 
\[
\arrayPointerP[\layerVar][\pathIndexToArrayIndex[\layerVar][\ell]] 
\]
and
\[
\arrayPointerC[\layerVar][\pathIndexToArrayIndex[\layerVar][\ell]] \; ,
\]
respectively.

Recall that at any given moment, some probability-pair and bit-pair arrays from our bank might be used by multiple paths, while others may not be used by any. The value of $\arrayReferenceCount[\layerVar][s]$ denotes the number of paths currently using the array pointed to by $\arrayPointerP[\layerVar][s]$. Note that this is also the number of paths making use of $\arrayPointerC[\layerVar][s]$. The index $s$ is contained in the stack $\inactiveArrayIndices[\layerVar]$ iff $\arrayReferenceCount[\layerVar][s]$ is zero.

Now that we have discussed how the data structures are initialized, we continue and discuss the low-level functions by which paths are made active and inactive. We start by mentioning Algorithm~\ref{alg:assignInitialPath}, by which the initial path of the algorithm is assigned and allocated. In words, we choose a path index $\ell$ that is not currently in use (none of them are), and mark it as used. Then, for each layer $\layerVar$, we mark (through $\pathIndexToArrayIndex$) an index $s$ such that both $\arrayPointerP[\layerVar][s]$ and $\arrayPointerC[\layerVar][s]$ are allocated to the current path.  

\begin{algorithm}
\caption{assignInitialPath$()$}
\label{alg:assignInitialPath}
\KwOut{index $\ell$ of initial path}
\BlankLine
$\ell$ $\gets$ $\Pop(\inactivePathIndices)$\;
$\activePath[\ell]$ $\gets$ \True\;
\tcp{Associate arrays with path index}
\For{$\layerVar=0,1,\ldots,m$}
{
  $s$ $\gets$ $\Pop(\inactiveArrayIndices[\layerVar])$\;
  $\pathIndexToArrayIndex[\layerVar][\ell]$ $\gets$ $s$\;
  $\arrayReferenceCount[\layerVar][s]$ $\gets$ $1$ \;
}
\Return $\ell$ \;
\end{algorithm}

Algorithm~\ref{alg:clonePath} is used to clone a path --- the final step before splitting that path in two. The logic is very similar to that of Algorithm~\ref{alg:assignInitialPath}, but now we make the two paths share bit-arrays and probability arrays. 
\begin{algorithm}
\caption{clonePath$(\ell)$}
\label{alg:clonePath}
\KwIn{index $\ell$ of path to clone}
\KwOut{index $\ell'$ of copy}
\BlankLine
$\ell'$ $\gets$ $\Pop(\inactivePathIndices)$\;
$\activePath[\ell']$ $\gets$ \True\;
\tcp{Make $\ell'$ reference same arrays as $\ell$}
\For{$\layerVar=0,1,\ldots,m$}
{
  $s$ $\gets$ $\pathIndexToArrayIndex[\layerVar][\ell]$\;
  $\pathIndexToArrayIndex[\layerVar][\ell']$ $\gets$ $s$\;
  $\arrayReferenceCount[\layerVar][s]{+}{+}$\;
}
\Return $\ell'$ \;
\end{algorithm}

Algorithm~\ref{alg:killPath} is used to terminate a path, which is achieved by marking it as inactive. After this is done, the arrays marked as associated with the path must be dealt with as follows. Since the path is inactive, we think of it as not having any associated arrays, and thus all the arrays that were previously associated with the path must have their reference count decreased by one.  

\begin{algorithm}
\caption{killPath$(\ell)$}
\label{alg:killPath}
\KwIn{index $\ell$ of path to kill}
\tcp{Mark the path index $\ell$ as inactive}
$\activePath[\ell]$ $\gets$ \False\;
$\Push(\inactivePathIndices,\ell)$\;
\tcp{Disassociate arrays with path index}
\For{$\layerVar=0,1,\ldots,m$}
{
  $s$ $\gets$ $\pathIndexToArrayIndex[\layerVar][\ell]$\;
  $\arrayReferenceCount[\layerVar][s]{-}{-}$\;
  \If{$\arrayReferenceCount[\layerVar][s]=0$}
  {
    $\Push(\inactiveArrayIndices[\layerVar], s)$\;
  }
}
\end{algorithm}

The goal of all previously discussed low-level functions was essentially to enable the abstraction implemented by the functions $\getArrayPointerP$ and $\getArrayPointerC$. The function $\getArrayPointerP$ is called each time a higher-level function needs to access (either for reading or writing) the probability-pair array associated with a certain path $\ell$ and layer $\layerVar$. The implementation of $\getArrayPointerP$ is give in Algorithm~\ref{alg:getArrayPointerP}. There are two cases to consider: either the array is associated with more than one path or it is not. If it is not, then nothing needs to be done, and we return a pointer to the array. On the other hand, if the array is shared, we make a private copy for path $\ell$, and return a pointer to that copy. By doing so, we ensure that two paths will never write to the same array. The function $\getArrayPointerC$ is used in the same manner for bit-pair arrays, and has exactly the same implementation, up to the obvious changes. 

At this point, we remind the reader that we are deliberately sacrificing speed for simplicity. Namely, each such function is called either before reading or writing to an array, but the copy operation is really needed only before writing. 

\begin{algorithm}
\caption{getArrayPointer\_P$(\layerVar,\ell)$}
\label{alg:getArrayPointerP}
\KwIn{layer $\layerVar$ and path index $\ell$}
\KwOut{pointer to corresponding probability pair array}
\BlankLine
\tcp{$\getArrayPointerC(\layerVar,\ell)$ is defined identically, up to the obvious changes in lines \ref{alg:getArrayPointerP:copy} and \ref{alg:getArrayPointerP:return}}
\BlankLine
$s$ $\gets$ $\pathIndexToArrayIndex[\layerVar][\ell]$\;
\eIf{$\arrayReferenceCount[\layerVar][s]=1$}
{
  $s'$ $\gets$ $s$\;
}
{
  $s'$ $\gets$ $\Pop(\inactiveArrayIndices[\layerVar])$\;
  copy the contents of the array pointed to by $\arrayPointerP[\layerVar][s]$ into that pointed to by $\arrayPointerP[\layerVar][s']$\; \label{algline:PCchange1} \label{alg:getArrayPointerP:copy}
  $\arrayReferenceCount[\layerVar][s]{-}{-}$ \;
  $\arrayReferenceCount[\layerVar][s']$ $\gets$ $1$\;
  $\pathIndexToArrayIndex[\layerVar][\ell]$ $\gets$ $s'$\;
}

\Return $\arrayPointerP[\layerVar][s']$\; \label{alg:getArrayPointerP:return}
\end{algorithm}

We have now finished defining almost all of our low-level functions. At this point, we should specify the constraints one should follow when using them and what one can expect if these constraints are met. We start with the former.
\begin{defi}[Valid calling sequence]
\label{def:validSequence}
Consider a sequence $(f_t)_{t=0}^T$ of $T+1$ calls to the low-level functions implemented in Algorithms~\ref{alg:initializeDataStructures}--\ref{alg:getArrayPointerP}. We say that the sequence is \emph{valid} if the following traits hold.

\emph{Initialized:} The one and only index $t$ for which $f_t$ is equal to $\initializeDataStructures$ is $t=0$. The one and only index $t$ for which $f_t$ is equal to $\assignInitialPath$ is $t=1$.

\emph{Balanced:} For $1 \leq t \leq T$, denote the number of times the function $\clonePath$ was called up to and including stage $t$ as
\[
\numberCalls{t}{\clonePath} = \\ 
\mycount{\myset{1 \leq i \leq t \; : \; \mbox{$f_i$ is $\clonePath$}}} \; .
\]
Define $\numberCalls{t}{\killPath}$ similarly. Then, for every $1 \leq t \leq L$, we require that
\begin{equation}
\label{eq:balanced}
1 \leq \left(1 + \numberCalls{t}{\clonePath} - \numberCalls{t}{\killPath} \right) \leq L \; .
\end{equation}

\emph{Active:} We say that path $\ell$ is active at the end of stage $1 \leq t \leq T$ if the following two conditions hold. First, there exists an index $1 \leq i \leq t$ for which $f_i$ is either $\clonePath$ with corresponding output $\ell$ or $\assignInitialPath$ with output $\ell$. Second, there is no intermediate index $i < j \leq t$ for which $f_j$ is $\killPath$ with input $\ell$. For each $1 \leq t  < T$ we require that if $f_{t+1}$ has input $\ell$, then $\ell$ is active at the end of stage $t$. 
\end{defi}

We start by stating that the most basic thing one would expect to hold does indeed hold.
\begin{lemm}
\label{lemm:wellDefined}
Let $(f_t)_{t=0}^T$ be a valid sequence of calls to the low-level functions implemented in Algorithms~\ref{alg:initializeDataStructures}--\ref{alg:getArrayPointerP}. Then, the run is well defined: i) A ``pop'' operation is never carried out on a empty stack, ii) a ``push'' operation never results in a stack with more than $L$ elements, and iii) a ``read'' operation from any array defined in lines~\ref{alg:initializeDataStructures:activePathDef}--\ref{alg:initializeDataStructures:arrayReferenceCountDef} of Algorithm~\ref{alg:initializeDataStructures} is always preceded by a ``write'' operation to the same location in the array.
\end{lemm}
\begin{proof}
The proof boils-down to proving the following four statements concurrently for the end of each step $1 \leq t \leq T$, by induction on $t$.
\begin{description}
\item[I] A path index $\ell$ is active by Definition~\ref{def:validSequence} iff $\activePath[\ell]$ is true iff $\inactivePathIndices$ does not contain the index $\ell$.
\item[II] The bracketed expression in (\ref{eq:balanced}) is the number of active paths at the end of stage $t$.
\item[III] The value of   $\arrayReferenceCount[\layerVar][s]$ is positive iff the stack $\inactiveArrayIndices[\layerVar]$ does not contain the index $s$, and is zero otherwise.
\item[IV] The value of $\arrayReferenceCount[\layerVar][s]$ is equal to the number of active paths $\ell$ for which $\pathIndexToArrayIndex[\layerVar][\ell]=s$. 
\end{description}
\end{proof}

We are now close to formalizing the utility of our low-level functions. But first, we must formalize the concept of a descendant path. Let $(f_t)_{t=0}^T$ be a valid sequence of calls. Next, let $\ell$ be an active path index at the end of stage $1 \leq t < T$. Henceforth, let us abbreviate the ``phrase path index $\ell$ at the end of stage $t$'' by ``$[\ell,t]$''. We say that $[\ell',t+1]$ is a child of $[\ell,t]$ if i) $\ell'$ is active at the end of stage $t+1$, and ii) either $\ell'=\ell$ or $f_{t+1}$ was the $\clonePath$ operation with input $\ell$ and output $\ell'$. Likewise, we say that $[\ell',t']$ is a descendant of $[\ell,t]$  if $1 \leq t \leq t'$ and there is a (possibly empty) hereditary chain.

We now broaden our definition of a valid function calling sequence by allowing reads and writes to arrays. 

\emph{Fresh pointer:} consider the case where $t > 1$ and $f_t$ is either the $\getArrayPointerP$ or $\getArrayPointerC$ function with input $(\layerVar,\ell)$ and output $p$. Then, for valid indices $i$, we allow read and write operations to $p[i]$ after stage $t$ but only before any stage $t' > t$ for which $f_{t'}$ is either $\clonePath$ or $\killPath$.
 
Informally, the following lemma states that each path effectively sees a private set of arrays.
\begin{lemm}
\label{lemm:freshPointer}
Let $(f_t)_{t=0}^T$ be a valid sequence of calls to the low-level functions implemented in Algorithms~\ref{alg:initializeDataStructures}--\ref{alg:getArrayPointerP}. Assume the read/write operations between stages satisfy the ``fresh pointer'' condition.

Let the function $f_t$ be $\getArrayPointerP$ with input $(\layerVar,\ell)$ and output $p$. Similarly, for stage $t' \geq t$, let $f_{t'}$ be $\getArrayPointerP$ with input $(\layerVar,\ell')$ and output $p'$. Assume that $[\ell',t']$ is a descendant of $[\ell,t]$.

Consider a ``fresh pointer'' write operation to $p[i]$. Similarly, consider a ``fresh pointer'' read operation from $p'[i]$ carried out after the ``write'' operation. Then, assuming no intermediate ``write'' operations of the above nature, the value written is the value read.

A similar claim holds for $\getArrayPointerC$.
\end{lemm}

\begin{proof}
With the observations made in the proof of Lemma~\ref{lemm:wellDefined} at hand, a simple induction on $t$ is all that is needed.
\end{proof}

We end this section by noting that the function $\pathIndexInactive$ given in Algorithm~\ref{alg:pathIndexInactive} is simply a shorthand, meant to help readability later on.

\begin{algorithm}
\caption{pathIndexInactive$(\ell)$}
\label{alg:pathIndexInactive}
\KwIn{path index $\ell$}
\KwOut{$\True$ if path $\ell$ is active, and $\False$ otherwise}
\BlankLine
\eIf{$\activePath[\ell]$ = \True}
{
  \Return \False
}
{
  \Return \True
}
\end{algorithm}

\subsection{Mid-level functions}
In this section we introduce Algorithms~\ref{alg:listSC_calcP} and \ref{alg:listSC_updateC}, our new implementation of Algorithms \ref{alg:improvedSC_calcP} and \ref{alg:improvedSC_updateC}, respectively, for the list decoding setting.

\begin{algorithm}
\caption{recursivelyCalcP$(\layerVar, \phaseVar)$ \hfill\emph{list version}}
\label{alg:listSC_calcP}
\KwIn{layer $\layerVar$ and phase $\phaseVar$}
\BlankLine
\lIf{$\layerVar = 0$}{\Return \tcp*[h]{Stopping condition}}\;
set $\psi \gets \floor{\phaseVar/2}$\;
\tcp{Recurse first, if needed}
\lIf{$\phaseVar \mod 2 = 0$}{$\recursivelyCalcP(\layerVar-1, \psi)$}\; \label{alg:listSC_calcP:recursiveCall}
\tcp{Perform the calculation}
$\sigma \gets 0$ \;
\For{$\ell = 0,1,\ldots,L-1$}
{ \label{alg:listSC_calcP:startProbCalc}
  \If{ $\pathIndexInactive(\ell)$}{\Continue\;}
  $P_\layerVar$ $\gets$ $\getArrayPointerP(\layerVar,\ell)$\;
  $P_{\layerVar-1}$ $\gets$ $\getArrayPointerP(\layerVar-1,\ell)$\;
  $C_\layerVar$ $\gets$ $\getArrayPointerC(\layerVar,\ell)$\;

  \For{$\branchVar=0,1,\ldots,2^{m-\layerVar}-1$}
  {
    \eIf{$\phaseVar \mod 2 = 0$}
    {
      \tcp*[r]{apply Equation (\ref{eq:ArikanFirstRecursive})}
      \For{$u' \in \{0,1\}$}
      {
        $\arrayAndBranch{P}{\layerVar}{\branchVar}[u'] \gets \sum_{u''} \frac{1}{2}
         \arrayAndBranch{P}{\layerVar-1}{2\branchVar}[u' \xor u''] 
         \cdot \arrayAndBranch{P}{\layerVar-1}{2\branchVar+1}[u'']
        $\;
        $\sigma \gets \myrmax{\sigma,\arrayAndBranch{P}{\layerVar}{\branchVar}[u']}$ \; \label{alg:listSC_calcP:minus}
      }
    }
    (\tcp*[f]{apply Equation (\ref{eq:ArikanSecondRecursive})})
    { 
      set $u' \gets \arrayAndBranch{C}{\layerVar}{\branchVar}[0]$\;
      \For{$u'' \in \{0,1\}$}
      {
        $\arrayAndBranch{P}{\layerVar}{\branchVar}[u''] \gets \frac{1}{2}
        \arrayAndBranch{P}{\layerVar-1}{2\branchVar}[u' \xor u''] 
        \cdot \arrayAndBranch{P}{\layerVar-1}{2\branchVar+1}[u'']
        $\;
        $\sigma \gets \myrmax{\sigma, \arrayAndBranch{P}{\layerVar}{\branchVar}[u'']}$ \; \label{alg:listSC_calcP:plus} \label{alg:listSC_calcP:endProbCalc}
      }
    }
  }
}
\tcp{normalize probabilities}
\For{$\ell = 0,1,\ldots,L-1$} 
{ \label{alg:listSC_calcP:normalizationStart}
  \If{ $\pathIndexInactive(\ell)$}{\Continue\;}
  $P_\layerVar$ $\gets$ $\getArrayPointerP(\layerVar,\ell)$\;

  \For{$\branchVar=0,1,\ldots,2^{m-\layerVar}-1$}
  {
    \For{$u \in \{0,1\}$}
    {
       $\arrayAndBranch{P}{\layerVar}{\branchVar}[u] \gets  \arrayAndBranch{P}{\layerVar}{\branchVar}[u]/\sigma$ \; \label{alg:listSC_calcP:normalizationEnd}
    }
  }
}
\end{algorithm}

One first notes that our new implementations loop over all path indices $\ell$. Thus, our new implementations make use of the functions $\getArrayPointerP$ and $\getArrayPointerC$ in order to assure that the consistency of calculations is preserved, despite multiple paths sharing information. In addition, Algorithm~\ref{alg:improvedSC_calcP} contains code to normalize probabilities. The normalization is needed for a technical reason (to avoid floating-point underflow), and will be expanded on shortly.

We start out by noting that the ``fresh pointer'' condition we have imposed on ourselves indeed holds. To see this, consider first Algorithm~\ref{alg:listSC_calcP}. The key point to note is that neither the $\killPath$ nor the $\clonePath$ function is called from inside the algorithm. The same observation holds for Algorithm~\ref{alg:listSC_updateC}. Thus, the ``fresh pointer'' condition is met, and Lemma~\ref{lemm:freshPointer} holds.

\begin{algorithm}
\caption{recursivelyUpdateC$(\layerVar,\phaseVar)$\hfill\emph{list version}}
\label{alg:listSC_updateC}
\KwIn{layer $\layerVar$ and phase $\phaseVar$}
\Require{$\phaseVar$ is odd}
\BlankLine
set $C_\layerVar$ $\gets$ $\getArrayPointerC(\layerVar,\ell)$\;
set $C_{\layerVar-1}$ $\gets$ $\getArrayPointerC(\layerVar-1,\ell)$\;

set $\psi \gets \floor{\phaseVar/2}$\;
\For{$\ell = 0,1,\ldots,L-1$}
{
  \If{ $\pathIndexInactive(\ell)$}{\Continue\;}

  \For{$\branchVar=0,1,\ldots,2^{m-\layerVar}-1$}
  {
  $\arrayAndBranch{C}{\layerVar-1}{2\branchVar}[\psi \mod 2] \gets \arrayAndBranch{C}{\layerVar}{\branchVar}[0] \xor \arrayAndBranch{C}{\layerVar}{\branchVar}[1]$\;
  $\arrayAndBranch{C}{\layerVar-1}{2\branchVar+1}[\psi \mod 2] \gets \arrayAndBranch{C}{\layerVar}{\branchVar}[1]$ \;
  }
}
\If{$\psi \mod 2 = 1$}
{
  $\recursivelyUpdateC(\layerVar-1, \psi)$\;
}
\end{algorithm}

We now consider the normalization step carried out in lines~\ref{alg:listSC_calcP:normalizationStart}--\ref{alg:listSC_calcP:normalizationEnd} of Algorithm~\ref{alg:listSC_calcP}. Recall that a floating-point variable can not be used to hold arbitrarily small positive reals, and in a typical implementation, the result of a calculation that is ``too small'' will be rounded to $0$. This scenario is called an ``underflow''.

We now confess that all our previous implementations of SC decoders were prone to ``underflow''. To see this, consider line~\ref{alg:highLevel} in the outline implementation given in Algorithm~\ref{alg:highLevel:calcW}. Denote by $\sfbY$ and $\sfbU$ the random vectors corresponding to $\sfby$ and $\sfbu$, respectively. For $b \in \mysett{0,1}$ we have that
\begin{multline*}
W_m^{(\phaseVar)}(\sfby_0^{n-1}, \hat{\sfbu}_0^{\phaseVar-1}|b) = \\
2 \cdot \myprobb{\sfbY_0^{n-1} = \sfby_0^{n-1}, \sfbU_0^{\phaseVar-1} = \hat{\sfbu}_0^{\phaseVar-1}, \sfU_\phaseVar = b } \leq \\
2 \cdot \myprobb{\sfbU_0^{\phaseVar-1} = \hat{\sfbu}_0^{\phaseVar-1}, \sfU_\phaseVar = b } = 2^{-\phaseVar} \; .
\end{multline*}
Recall that $\phaseVar$ iterates from $0$ to $n-1$. Thus, for codes having length greater than some small constant, the comparison in line~\ref{alg:highLevel} of Algorithm~\ref{alg:highLevel:calcW} ultimately becomes meaningless, since both probabilities are rounded to $0$. The same holds for all of our previous implementations.

Luckily, there is a simple fix to this problem. After the probabilities are calculated in lines~\ref{alg:listSC_calcP:startProbCalc}--\ref{alg:listSC_calcP:endProbCalc} of Algorithm~\ref{alg:listSC_calcP}, we normalize\footnote{This correction does not assure us that underflows will not occur. However, now, the probability of a meaningless comparison due to underflow will be extremely low.} the highest probability to be $1$ in lines~\ref{alg:listSC_calcP:normalizationStart}--\ref{alg:listSC_calcP:normalizationEnd}. 

We claim that apart for avoiding underflows, normalization does not alter our algorithm. The following lemma formalizes this claim.

\begin{lemm}
\label{lemm:normalization}
Assume that we are working with ``perfect'' floating-point numbers. That is, our floating-point variables are infinitely accurate and do not suffer from underflow/overflow. Next, consider a variant of Algorithm~\ref{alg:listSC_calcP}, termed Algorithm~\ref{alg:listSC_calcP}', in which just before line \ref{alg:listSC_calcP:normalizationStart} is first executed, the variable $\sigma$ is set to $1$. That is, effectively, there is no normalization of probabilities in Algorithm~\ref{alg:listSC_calcP}'.

Consider two runs, one of Algorithm~\ref{alg:listSC_calcP} and one of Algorithm~\ref{alg:listSC_calcP}'. In both runs, the input parameters to both algorithms are the same. Moreover, assume that in both runs, the state of the auxiliary data structures is the same, apart for the following.

Recall that our algorithm is recursive, and let $\layerVar_0$ be the first value of the variable $\layerVar$ for which line~\ref{alg:listSC_calcP:startProbCalc} is executed. That is, $\layerVar_0$ is the layer in which (both) algorithms do not perform preliminary recursive calculations. Assume that when we are at this base stage, $\layerVar = \layerVar_0$, the following holds: the values read from $P_{\layerVar-1}$ in lines~\ref{alg:listSC_calcP:minus} and \ref{alg:listSC_calcP:plus} in the run of Algorithm~\ref{alg:listSC_calcP} are a multiple by $\alpha_{\layerVar-1}$ of the corresponding values read in the run of Algorithm~\ref{alg:listSC_calcP}'. Then, for every $\layerVar \geq \layerVar_0$, there exist a constant $\alpha_{\layerVar}$ such that the values written to $P_{\layerVar}$ in line~\ref{alg:listSC_calcP:normalizationEnd} in the run of Algorithm~\ref{alg:listSC_calcP} are a multiple by $\alpha_{\layerVar}$ of the corresponding values written by Algorithm~\ref{alg:listSC_calcP}'.
\end{lemm}

\begin{proof}
For the base case $\layerVar = \layerVar_0$ we have by inspection that the constant $\alpha_\layerVar$ is simply $(\alpha_{\layerVar-1})^2$, divided by the value of $\sigma$ after the main loop has finished executing in Algorithm~\ref{alg:listSC_calcP}. The claim for a general $\layerVar$ follows by induction.
\end{proof}

\subsection{High-level functions}
We now turn our attention to the high-level functions of our algorithm. Consider the topmost function, the main loop given in Algorithm~\ref{alg:list_main}. We start by noting that by lines \ref{alg:list_main:initializeDataStructures} and \ref{alg:list_main:assignInitialPath}, we have that condition  ``initialized'' in Definition~\ref{def:validSequence} is satisfied. Also, for the inductive basis, we have that condition ``balanced'' holds for $t=1$ at the end of line \ref{alg:list_main:assignInitialPath}. Next, notice that lines \ref{alg:list_main:P0start}--\ref{alg:list_main:P0end} are in-line with our ``fresh pointer'' condition.

The main loop, lines \ref{alg:list_main:mainLoopStart}--\ref{alg:list_main:mainLoopEnd}, is the analog of the main loop in Algorithm~\ref{alg:improvedSC_main}. After the main loop has finished, we pick (in lines \ref{alg:list_main:pickCodewordStart}--\ref{alg:list_main:pickCodewordEnd}) the most likely codeword from our list and return it.

\begin{algorithm}
\KwIn{the received vector $\sfby$ and a list size $L$ as a global}
\KwOut{a decoded codeword $\hat{\sfbc}$}
\BlankLine
\caption{SCL decoder, main loop}
\label{alg:list_main}
\tcp{Initialization}
$\initializeDataStructures()$\; \label{alg:list_main:initializeDataStructures}
$\ell \gets \assignInitialPath()$\; \label{alg:list_main:assignInitialPath}
$P_0 \gets \getArrayPointerP(0,\ell)$\; \label{alg:list_main:P0start}
\For{$\branchVar = 0,1,\ldots, n-1$}{
  set $\arrayAndBranch{P}{0}{\branchVar}[0] \gets W(\sfy_\branchVar|0)$, $\arrayAndBranch{P}{0}{\branchVar}[1] \gets W(\sfy_\branchVar|1)$ \; \label{alg:list_main:P0end}
  } 
\tcp{Main loop} 
\For{$\phaseVar = 0,1,\ldots, n-1$} 
{ \label{alg:list_main:mainLoopStart}
$\recursivelyCalcP(m, \phaseVar)$\; \label{alg:list_main:recursivelyCalcP}  

     \eIf{$\sfu_\phaseVar$ is frozen}
     {
       $\continuePathsFrozenBit(\phaseVar)$\;
     }
     { 
       $\continuePathsUnfrozenBit(\phaseVar)$\;
     }

  \If{$\phaseVar \mod 2 = 1$}
  {
    \recursivelyUpdateC$(m, \phaseVar)$\; \label{alg:list_main:mainLoopEnd} \label{alg:list_main:recursivelyUpdateC}
  }
} 
\tcp{Return the best codeword in the list} 
$\ell \gets \findMostProbablePath()$ \; \label{alg:list_main:pickCodewordStart}

set $C_0 \gets \getArrayPointerC(0,\ell)$\; 

\Return $\hat{\sfbc}=\left(\arrayAndBranch{C}{0}{\branchVar}[0]\right)_{\branchVar=0}^{n-1}$\; \label{alg:list_main:pickCodewordEnd}
\end{algorithm}

We now expand on Algorithms~\ref{alg:list_continueFrozen} and \ref{alg:list_continueUnfrozen}. Algorithm~\ref{alg:list_continueFrozen} is straightforward: it is the analog of line~\ref{alg:improvedSC_main:frozenSet} in Algorithm~\ref{alg:improvedSC_main}, applied to all active paths.

\begin{algorithm}
\caption{continuePaths\_FrozenBit$(\phaseVar)$}
\label{alg:list_continueFrozen}
\KwIn{phase $\phaseVar$}
\BlankLine

\For{$\ell = 0,1,\ldots,L-1$}
{
  \lIf{ $\pathIndexInactive(\ell)$}{\Continue\;}
  $C_m \gets \getArrayPointerC(m,\ell)$\;
  set $\arrayAndBranch{C}{m}{0}[\phaseVar \mod 2]$ to the frozen value of $\sfu_\phaseVar$\;
}
\end{algorithm}

\begin{algorithm}
\caption{continuePaths\_UnfrozenBit$(\phaseVar)$}
\label{alg:list_continueUnfrozen}
\KwIn{phase $\phaseVar$}
\BlankLine

$\probForks \gets$ \New 2-D float array of size $L \times 2$ \;
$i \gets 0$\;

\tcp{populate \probForks}
\For{$\ell = 0,1,\ldots,L-1$}
{
  \eIf{ $\pathIndexInactive(\ell)$}
  {
\probForks[$\ell$][$0$] $\gets$ $ \; -1$\;
\probForks[$\ell$][$1$] $\gets$ $ \; -1$\;
}{

   $P_m \gets \getArrayPointerP(m,\ell)$\;

\probForks[$\ell$][$0$] $\gets \arrayAndBranch{P}{m}{0}[0]$\;
\probForks[$\ell$][$1$] $\gets \arrayAndBranch{P}{m}{0}[1]$\;
$i \gets i + 1$\;
}
}
$\rho \gets \min(2i,L)$\;
$\contForks \gets$ \New 2-D boolean array of size $L \times 2$ \;
\tcp{The following is possible in $O(L)$ time}
populate $\contForks$ such that $\contForks[\ell][b]$ is $\True$ iff \probForks[$\ell$][$b$] is one of the $\rho$ largest entries in $\probForks$ (and ties are broken arbitrarily) \; \label{alg:list_continueUnfrozen:pivot}

\tcp{First, kill-off non-continuing paths}
\For{$\ell = 0,1,\ldots,L-1$}
{
  \If{ $\pathIndexInactive(\ell)$}{\Continue\;}
  \If{$\contForks[\ell][0] = \False \; \myand \; \contForks[\ell][1] = \False$}
  {
    $\killPath(\ell)$
  }
}

\tcp{Then, continue relevant paths, and duplicate if necessary}
\For{$\ell = 0,1,\ldots,L-1$}
{
  \If(\tcp*[h]{both forks are bad, or invalid}){$\contForks[\ell][0] = \False \; \myand \; \contForks[\ell][1] = \False$} 
  {
    \Continue  
  }

  $C_m \gets \getArrayPointerC(m,\ell)$\;

  \eIf(\tcp*[h]{both forks are good}){$\contForks[\ell][0] = \True \; \myand \; \contForks[\ell][1] = \True$} 
  {
     set $\arrayAndBranch{C}{m}{0}[\phaseVar \mod 2] \gets 0$\;
     $\ell' \gets \clonePath(\ell)$\;
     $C_m \gets \getArrayPointerC(m,\ell')$\;
     set $\arrayAndBranch{C}{m}{0}[\phaseVar \mod 2] \gets 1$\;
  }
  (\tcp*[h]{exactly one fork is good})
  {
    \eIf{$\contForks[\ell][0] = \True$}
    {
      set $\arrayAndBranch{C}{m}{0}[\phaseVar \mod 2] \gets 0$\;
    }
    {
      set $\arrayAndBranch{C}{m}{0}[\phaseVar \mod 2] \gets 1$\;
    }
  }

}
\end{algorithm}

Algorithm~\ref{alg:list_continueUnfrozen} is the analog of lines~\ref{alg:firstImplementation_main:unfrozenStart}--\ref{alg:firstImplementation_main:unfrozenEnd} in Algorithm~\ref{alg:improvedSC_main}. However, now, instead of choosing the most likely fork out of $2$ possible forks, we must typically choose the $L$ most likely forks out of $2L$ possible forks. The most interesting line is ~\ref{alg:list_continueUnfrozen:pivot}, in which the best $\rho$ forks are marked. Surprisingly\footnote{The $O(L)$ time result is rather theoretical. Since $L$ is typically a small number, the fastest way to achieve our selection goal would be through simple sorting.}, this can be done in $O(L)$ time \cite[Section 9.3]{CLRS:01b}. After the forks are marked, we first kill the paths for which both forks are discontinued, and then continue paths for which one or both are the forks are marked. In case of the latter, the path is first split. Note that we must first kill paths and only then split paths in order for the ``balanced'' constraint (\ref{eq:balanced}) to hold. Namely, this way, we will not have more than $L$ active paths at a time.

The point of Algorithm~\ref{alg:list_continueUnfrozen} is to prune our list and leave only the $L$ ``best'' paths. This is indeed achieved, in the following sense. At stage $\phaseVar$ we would like to rank each path according the the probability
\[
W_m^{(\phaseVar)}(\sfby_0^{n-1}, \hat{\sfbu}_{0}^{\phaseVar-1}|\hat{\sfu}_\phaseVar) \; . 
\]
By (\ref{eq:probArray}) and (\ref{eq:PReplace}), this would indeed by the case if our floating point variables were ``perfect'', and the normalization step in lines~\ref{alg:listSC_calcP:normalizationStart}--\ref{alg:listSC_calcP:normalizationEnd} of Algorithm~\ref{alg:listSC_calcP} were not carried out. By Lemma~\ref{lemm:normalization}, we see that this is still the case if normalization is carried out.

The last algorithm we consider in this section is Algorithm~\ref{alg:findMostProbablePath}. In it, the most probable path is selected from the final list. As before, by (\ref{eq:probArray})--(\ref{eqn:BCReplace}) and Lemma~\ref{lemm:normalization}, the value of $P_m[0][C_m[0][1]]$ is simply
\[
W_m^{(n-1)}(\bfgenericReceived_0^{n-1}, \hat{\bfu}_{0}^{n-2}|\hat{u}_{n-1}) = \frac{1}{2^{n-1}} \cdot P(  \bfy_0^{n-1} |\hat{\bfu}_0^{n-1} ) \; ,
\]
up to a normalization constant.

\begin{algorithm}
\caption{findMostProbablePath$()$}
\label{alg:findMostProbablePath}
\KwOut{the index $\ell'$ of the most probable path}
\BlankLine
$\ell'$ $\gets$ $0$, $p'$ $\gets$ $0$\;
\For{$\ell = 0,1,\ldots,L-1$}
{
  \If{ $\pathIndexInactive(\ell)$}{\Continue\;}
  $C_m \gets \getArrayPointerC(m,\ell)$\;
  $P_m \gets \getArrayPointerP(m,\ell)$\;

  \If{$p' < P_m[0][C_m[0][1]]$}
  {
    $\ell'$ $\gets$ $\ell$, $p'$ $\gets$ $P_m[0][C_m[0][1]]$
  }

}
\Return $\ell'$ \;
\end{algorithm}

We now prove our two main result.

\begin{theo}
The space complexity of the SCL decoder is $O(L \cdot n)$.
\end{theo}
\begin{proof}
All the data-structures of our list decoder are allocated in Algorithm~\ref{alg:initializeDataStructures}, and it can be checked that the total space used by them is $O(L \cdot n)$. Apart from these, the space complexity needed in order to perform the selection operation in line~\ref{alg:list_continueUnfrozen:pivot} of Algorithm~\ref{alg:list_continueUnfrozen} is $O(L)$. Lastly, the various local variables needed by the algorithm take $O(1)$ space, and the stack needed in order to implement the recursion takes $O(\log n)$ space.
\end{proof}

\begin{theo}
The running time of the SCL decoder is $O(L \cdot n \log n)$.
\end{theo}
\begin{proof}
Recall that by our notation $m = \log n$. The following bottom-to-top table summarizes the running time of each function. The notation $O_\Sigma$ will be explained shortly.
\begin{center}
\begin{tabular}{ll}
\hline 
function & running time\\
\hline \hline
$\initializeDataStructures()$ & $O(L \cdot m)$ \\
$\assignInitialPath()$ & $O(m)$ \\
$\clonePath(\ell)$ & $O(m)$ \\
$\killPath(\ell)$ & $O(m)$ \\
$\getArrayPointerP(\layerVar,\ell)$ & $O(2^{m-\layerVar})$ \\
$\getArrayPointerC(\layerVar,\ell)$ & $O(2^{m-\layerVar})$ \\
$\pathIndexInactive(\ell)$ & $O(1)$ \\
\hline
$\recursivelyCalcP(m, \cdot)$ & $O_\Sigma(L \cdot m \cdot n)$ \\
$\recursivelyUpdateC(m, \cdot)$ & $O_\Sigma(L \cdot m \cdot n)$ \\
\hline
$\continuePathsFrozenBit(\phaseVar)$ & $O(L)$ \\
$\continuePathsFrozenBit(\phaseVar)$ & $O(L \cdot m)$ \\
$\findMostProbablePath$ & $O(L)$ \\
\hline
SCL decoder & $O(L \cdot m \cdot n )$ \\
\hline
\end{tabular}
\end{center}
The first $7$ functions in the table, the low-level functions, are easily checked to have the stated running time. Note that the running time of $\getArrayPointerP$ and $\getArrayPointerC$ is due to the copy operation in line~\ref{alg:getArrayPointerP:copy} of Algorithm~\ref{alg:getArrayPointerP:copy} applied to an array of size $O(2^{m-\layerVar})$. Thus, as was previously mentioned, reducing the size of our arrays has helped us reduce the running time of our list decoding algorithm.

Next, let us consider the two mid-level functions, namely, $\recursivelyCalcP$ and $\recursivelyUpdateC$. The notation
\[
\recursivelyCalcP(m, \cdot) \in O_\Sigma(L \cdot m \cdot n)
\]
means that total running time of the $n$ function calls
\[
\recursivelyCalcP(m, \phaseVar) \; , \quad 0 \leq \phaseVar < 2^m
\]
is $O(L \cdot m \cdot n)$. To see this, denote by $f(\lambda)$ the total running time of the above with $m$ replaced by $\layerVar$. By splitting the running time of Algorithm~\ref{alg:listSC_calcP} into a non-recursive part and a recursive part, we have that for $\layerVar > 0$
\[
f(\lambda) = 2^\layerVar \cdot O(L \cdot 2^{m-\layerVar}) + f(\lambda-1) \; .
\]
Thus, it easily follows that
\[
f(m) \in O(L \cdot m \cdot 2^m) = O(L \cdot m \cdot n) \; .
\]
In essentially the same way, we can prove that the total running time of the $\recursivelyUpdateC(m, \phaseVar)$ over all $2^{n-1}$ valid (odd) values of $\phaseVar$ is $O(m \cdot n)$. Note that the two mid-level functions are invoked in lines~\ref{alg:list_main:recursivelyCalcP} and \ref{alg:list_main:recursivelyUpdateC} of Algorithm~\ref{alg:list_main}, on all valid inputs.

The running time of the high-level functions is easily checked to agree with the table.

\end{proof}

\section{Modified polar codes}
\label{sec:listCRC}

\begin{figure}
\begin{center}
\includegraphics[scale=0.50]{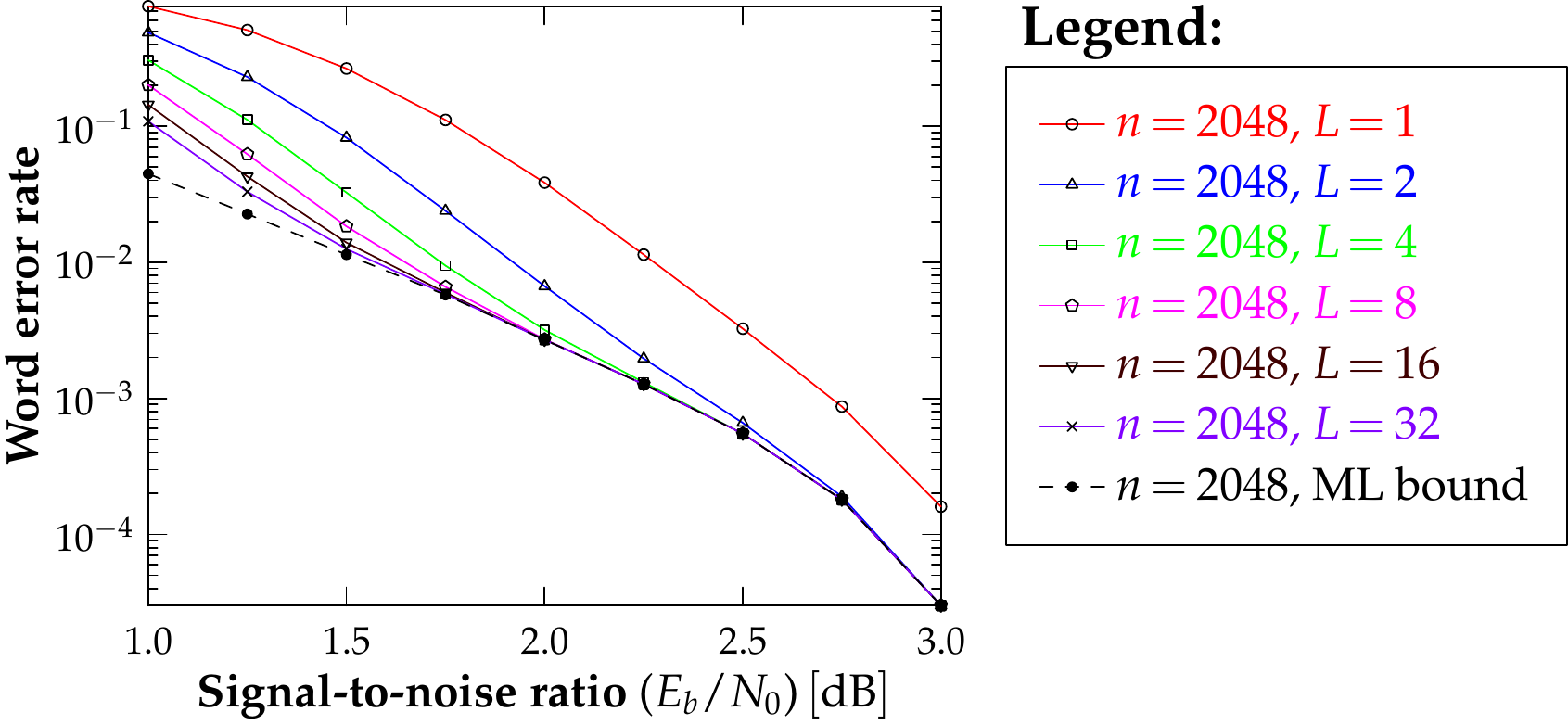}

\vspace{10pt}

\includegraphics[scale=0.50]{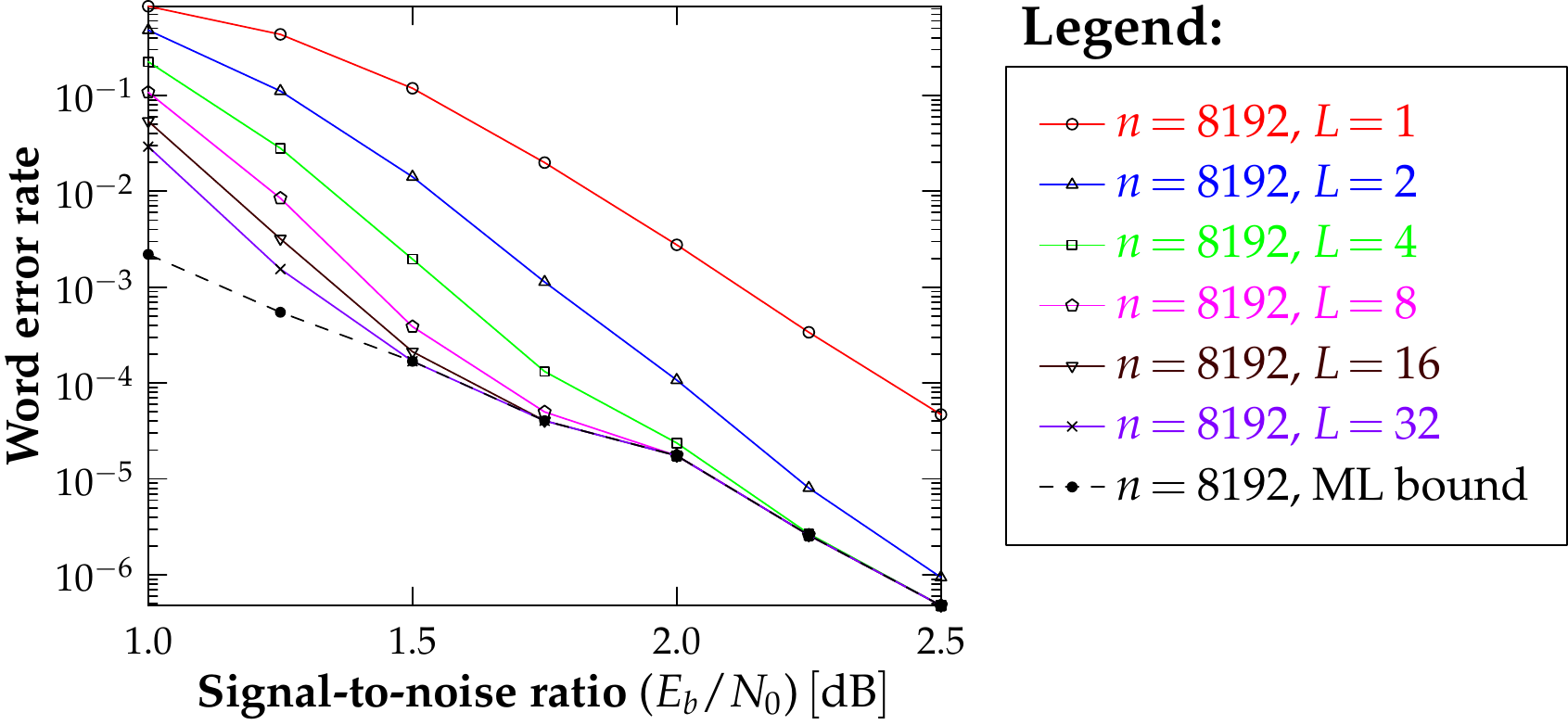}
\caption{Word error rate of a length $n=2048$ (top) and $n=8192$ (bottom) rate $1/2$ polar code optimized for SNR=$2$ dB under various list sizes. Code construction was carried out via the method proposed in \cite{TalVardy:11a}.}
\label{fig:WERNoCRC}
\end{center}
\end{figure}

The plots in Figure~\ref{fig:WERNoCRC} were obtained by simulation. The performance of our decoder for various list sizes is given by the solid lines in the figure. As expected, we see that as the list size $L$ increases, the performance of our decoder improves. We also notice a diminishing-returns phenomenon in terms of increasing the list size. The reason for this turns out to be simple.

The dashed line, termed the ``ML bound'' was obtained as follows. During our simulations for $L=32$, each time a decoding failure occurred, we checked whether the decoded codeword was more likely than the transmitted codeword. That is, whether $W(\sfby|\hat{\sfbc}) > W(\sfby|\sfbc)$. If so, then the optimal ML decoder would surely  misdecode $\sfby$ as well. The dashed line records the frequency of the above event, and is thus a lower-bound on the error probability of the ML decoder. Thus, for an SNR value greater than about $1.5 \; \mathrm{dB}$, Figure~\ref{fig:WER} suggests that we have an essentially optimal decoder when $L=32$.

Can we do even better? At first, the answer seems to be an obvious ``no'', at least for the region in which our decoder is essentially optimal. However, it turns out that if we are willing to accept a small change in our definition of a polar code, we can dramatically improve performance.

During simulations we noticed that often, when a decoding error occurred, the path corresponding to the transmitted codeword was a member of the final list. However, since there was a more likely path in the list, the codeword corresponding to that path was returned, which resulted in a decoding error. Thus, if only we had a ``genie'' to tell as at the final stage which path to pick from our list, we could improve the performance of our decoder.

Luckily, such a genie is easy to implement. Recall that we have $k$ unfrozen bits that we are free to set. Instead of setting all of them to information bits we wish to transmit, we employ the following simple concatenation scheme. For some small constant $r$, we set the first $k-r$ unfrozen bits to information bits. The last $r$ unfrozen bits will hold the $r$-bit CRC \cite[Section 8.8]{PetersonWeldon:72b} value\footnote{A binary linear code having a corresponding $k \times r$  parity-check matrix constructed as follows will do just as well. Let the  the first $k-r$ columns be chosen at random and the last $r$ columns be equal to the identity matrix.} of the first $k-r$ unfrozen bits. Note this new encoding is a slight variation of our polar coding scheme. Also, note that we incur a penalty in rate, since the rate of our code is now $(k-r)/n$ instead of the previous $k/n$.

What we have gained is an approximation to a genie: at the final stage of decoding, instead of calling the function $\findMostProbablePath$ in Algorithm~\ref{alg:findMostProbablePath}, we can do the following. A path for which the CRC is invalid can not correspond to the transmitted codeword. Thus, we refine our selection as follows. If at least one path has a correct CRC, then we remove from our list all paths having incorrect CRC and then choose the most likely path. Otherwise, we select the most likely path in the hope of reducing the number of bits in error, but with the knowledge that we have at least one bit in error. 

Figures~\ref{fig:WER} and \ref{fig:WiMax} contain a comparison of decoding performance between the original polar codes and the slightly tweaked version presented in this section. A further improvement in bit-error-rate (but not in block-error-rate) is attained when the decoding is performed systematically \cite{Arikan:11p}. The application of systematic polar-coding to a list decoding setting is attributed to \cite{SarkisGross:11pc}.

\twobibs{
\bibliographystyle{IEEEtran}
\bibliography{../../mybib}
}
{

}
\end{document}